\documentclass[journal,twocolumn,font=10]{IEEEtran}
\usepackage[T1]{fontenc}
\usepackage{amssymb}
\usepackage[utf8]{inputenc}

\usepackage[english]{babel}
\usepackage{amsmath}				
\usepackage{comment}				
\usepackage[T1]{fontenc}
\usepackage{mathtools}
\usepackage{algcompatible,amsmath}
\usepackage[ruled,vlined]{algorithm2e}
\usepackage{graphicx}
\usepackage{multicol}
\usepackage{babel,blindtext}
\usepackage{caption}
\usepackage{afterpage}
\usepackage[noend]{algpseudocode}
\usepackage{color}
\usepackage{etoolbox}
\usepackage{url}
\usepackage{subcaption,siunitx,booktabs}
\usepackage{siunitx,booktabs}
\usepackage{placeins}
\usepackage{bbm}
\usepackage{soul}
\usepackage{matlab-prettifier}
\usepackage{listings}
\usepackage[noadjust]{cite}

\usepackage{matlab-prettifier}

\makeatletter
\patchcmd{\@makecaption}
{\scshape}
{}
{} 
{}
\makeatother
\usepackage{graphicx}

\usepackage{amsthm}

\newtheorem{proposition}{Proposition}

\newtheorem*{proof*}{Proof}

\newtheorem{remark}{Remark}	

\newcommand{\norm}[1]{\left\lVert#1\right\rVert}
\usepackage{xparse}

\NewDocumentCommand{\INTERVALINNARDS}{ m m }{
	#1 {,} #2
}
\NewDocumentCommand{\interval}{ s m >{\SplitArgument{1}{,}}m m o }{
	\IfBooleanTF{#1}{
		\left#2 \INTERVALINNARDS #3 \right#4
	}{
		\IfValueTF{#5}{
			#5{#2} \INTERVALINNARDS #3 #5{#4}
		}{
			#2 \INTERVALINNARDS #3 #4
		}
	}
}

\IEEEoverridecommandlockouts
\makeatletter
\let\origIEEEPARstart\IEEEPARstart
\renewcommand{\IEEEPARstart}[3][1.1]{%
	\def\@IEEEPARstartDROPDEPTH{#1\baselineskip}%
	\origIEEEPARstart{#2}{#3}%
}
\makeatother

\ifCLASSINFOpdf

\else

\fi

\begin{document}
	\title{Data-Driven Adaptive Network Slicing for Multi-Tenant Networks}
	
	\author{Navid~Reyhanian and Zhi-Quan Luo, \textit{Fellow, IEEE}
		\thanks{This paper was accepted to be presented in part at IEEE International Conference on Acoustics, Speech and Signal Processing (ICASSP),  June 6–11, 2021 \cite{Reyhanian2021icassp}.}
		\thanks{N. Reyhanian is with the Department
			of Electrical and Computer Engineering, University of Minnesota, Minneapolis,
			MN, 55455 USA (e-mail: navid@umn.edu).}
		\thanks{Z.-Q. Luo is with Shenzhen Research Institute of Big Data, and The Chinese University of Hong Kong, Shenzhen, China (e-mail: luozq@cuhk.edu.cn).}}
	

	\maketitle
	
	\begin{abstract}
		Network slicing to support multi-tenancy plays a key role in improving the performance of 5G networks. In this paper, we propose a two time-scale framework for the reservation-based network slicing in the backhaul and Radio Access Network (RAN). In the proposed two time-scale scheme, a subset of network slices is activated via a novel sparse optimization framework in the long time-scale with the goal of maximizing the expected utilities of tenants while in the short time-scale the activated slices are reconfigured according to the time-varying user traffic and channel states. Specifically, using the statistics from users and channels and also considering the expected utility from serving users of a slice and the reconfiguration cost, we formulate a sparse optimization problem to update the configuration of a slice resources such that the maximum isolation of reserved resources is enforced. The formulated optimization problems for long and short time-scales are non-convex and difficult to solve. We use the $\ell_q$-norm, $0<q<1$, and group LASSO regularizations to iteratively find convex approximations of the optimization problems. We propose a Frank-Wolfe algorithm to iteratively solve approximated problems in long time-scales. To cope with the dynamical nature of traffic variations, we propose a fast, distributed algorithm to solve the approximated optimization problems in short time-scales. Simulation results demonstrate the performance of our approaches relative to optimal solutions and the existing state of the art method.

	\end{abstract}


	\IEEEpeerreviewmaketitle

	\section{Introduction}\label{sec:intro}
	Recently, significant attention has been placed on network slicing as a key element for enabling flexibility and programmability in 5G mobile networks \cite{alliance2016description,8744260}. The network \textit{slice} is a logical, virtualized end-to-end network provided to each tenant to support the demands of users. Each slice consists of link capacities in the backhaul and transmission resources, e.g., bandwidth, in the Radio Access Network (RAN), while each \textit{tenant} may own multiple network slices and use their reserved resources to serve users. 
	Network slices are dynamically activated, reconfigured and deactivated by the control center of the infrastructure provider. 
	
	
	The life-cycle management of network slices includes the design and creation phase, orchestration and activation phase, and the optimization and reconfiguration phase \cite{8004166}. In the design and creation phase, network slices are conceptually constructed based on user demands. In the orchestration and activation phase, network slices are installed on the shared physical infrastructure and user traffic starts flowing through the slice. Finally, in the optimization and reconfiguration phase, the performance of network slices is monitored, and based on traffic variations, network slices are reconfigured to maintain the Quality of Service (QoS) requirements of users. 
	
	The requirements of network slicing include scalability, flexibility, isolation, and efficient end-to-end orchestration \cite{khan2020network}.
	The scalability enables the network to efficiently adapt to provide a wide variety of applications. The flexibility allows slice reconfiguration to improve the concert of resources across different parts of the network to provide a particular service \cite{rost2017network}. In spite of the many benefits, the slice reconfiguration incurs costs and service interruptions. Therefore, slices are reconfigured only if the new configurations significantly improve network performance \cite{wang2019reconfiguration}. The isolation of reserved resources for different slices over shared physical infrastructure enables independent management of network slices and ensures that a rapid increase in the number of users, a slice failure or a security attack on one slice do not affect other slices \cite{li2017network}. The end-to-end orchestration helps to coordinate multiple system facets to maintain the QoS requirements of users. Although network slicing in the backhaul \cite{8454739,baumgartner2017optimisation,paris2016controlling} and in RAN \cite{7982693,8642526} are traditionally studied separately, only the end-to-end orchestration can ensure a robust resource allocation and reliable network performance. 
	
	To fully leverage the benefits of network slicing in 5G, it is necessary to dynamically reconfigure slices and allocate resources in a flexible data-driven manner. The data-driven methods adapt slices to traffic variations and channel states \cite{pozza2020reconfiguring}. Data-driven adaptive network slicing enables efficient and flexible allocation of resources to better support the QoS requirements of users. The QoS attributes include minimum data rate, maximum rate-loss, reliability, and security. Depending on the  QoS demands, users can be assigned to different slices. 
	
	\subsection{Related Work}
	The combinatorial problem of network slice activation is studied in \cite{8698758}, where the incoming traffic is supposed to follow a Poisson distribution. In \cite{8698758}, a network slice is kept  active based on the gained utility from serving the most recent demand without considering the predictions for the future traffic. The activation cost, which is neither convex nor continuous, is not considered. A heuristic approach for network slice activation is proposed in \cite{9165512}, which only considers the real-time status of network slices. In practice, due to the large cost of a slice activation, network slices are operated based on long-term future traffic rather than the instantaneous demand. No existing paper in the literature considers the mathematical problem of slice activation in an end-to-end network for unseen user demands.
	
	Caballero \textit{et al.} studied the allocation of network transmission resources among several tenants in \cite{7982693}, where transmission rates to users are maximized based on a weighted proportional fairness. The number of users served by each tenant is assumed to be a random number in \cite{7982693}. An optimal dynamic resource allocation problem is formulated and shown to be NP-hard.  A greedy approach is proposed in \cite{7982693} to solve the problem without a proof of convergence. To find the resource allocations for different slices of a network, a game-theoretic approach which maximizes an $\alpha-$fairness is proposed in \cite{8642526}. Static and dynamic pricing frameworks are proposed in \cite{8642526} to allocate base station transmission resources to slices, where each slice has a non-cooperative, strategic behavior. 
	The drawbacks of  \cite{8642526} include neglecting traffic statistics and failing to address the joint resource allocation in the backhaul and RAN. An adaptive forecasting and bandwidth allocation for cyclic demands in service-oriented networks is proposed in \cite{oliveira2021adaptive} without accounting for the reconfiguration cost. The user-slice association problem to find the best  slice to support requirements of each user is studied in our previous work \cite{Reyhanian2}.
	
	An optimization framework for flexible inter-tenant resource
	sharing with transmission power control is proposed in \cite{8717730} to improve network capacity and the utilization of base station resources, where interference levels are controllable. However, achievable rates of channels are assumed to be deterministic, which is an impractical assumption as the achievable rates are random in wireless channels \cite{1210731,musavian2015effective,4411539}. Network slicing algorithms for slice recovery and reconfiguration 
	under stochastic demands in service-oriented networks are studied in \cite{8454739}. Wang \textit{et al.} use $\ell_1$-norm  to promote sparsity in slice reconfigurations in \cite{wang2019reconfiguration}. To tackle the non-differentiability of $\ell_1$-norm, affine constraints that allow limited slice variations are considered instead of $\ell_1$-norm in \cite{wang2019reconfiguration}. The statistics of the demand are not used in \cite{wang2019reconfiguration} for network slicing.  In \cite{pozza2020reconfiguring}, Pozza \textit{et al.} propose a heuristic divide-and-conquer approach for finding a sequence of feasible
	solutions to reconfigure a slice made of a chain of network functions under bandwidth and latency constraints, while RAN is not considered in the formulations. 
	To tackle the high computational complexity caused by the large number of  variables in the slice reconfiguration problem, a heuristic depth-first-search algorithm is proposed in \cite{wei2020network} to find a set of possible reconfigurations, and a reinforcement learning approach is used to explore the multi-dimensional discrete action space. An end-to-end network slicing method for 5G networks without imposing constraints on available resources is proposed in \cite{yang2020data}, where the slice status is monitored and necessarily reconfigured in an online fashion via a heuristic approach. However, reconfiguring end-to-end slices based on instantaneous demands significantly increases the reconfiguration costs.
	
	Unlike most papers, e.g., \cite{8642526,7982693,8717730}, that slice the network based on traffic variations in single time-scale, a few recent papers, e.g., \cite{zhang2019deep,9057480,zhang2018joint}, propose two-time scale frameworks to improve network management and efficiency of resource allocations. In \cite{9057480}, a two-time-scale resource management
	scheme for network slicing in cloud RAN is proposed. Zhang \textit{et al.} propose a long time-scale inter-slice resource reservation for slices and a short time-scale intra-slice resource allocation in \cite{9057480}. However, the formulation of \cite{9057480} does not consider the slice reconfiguration cost, which is non-continuous and non-convex \cite{wang2019reconfiguration}, nor the resource reservation in the backhaul. Thus, \cite{9057480} fails to provide end-to-end QoS guarantees to users. A multi-time-scale decentralized online orchestration of software-defined networks is studied in \cite{8468185}, where a set of network controllers are activated based on the temporal and spatial variations in traffic requests. The slice activation and reconfiguration in a two time-scale framework are not studied in the existing literature; these topics comprise the main focus of this work.

	\subsection{Our Contributions}
	In this paper, we propose a two-time-scale resource management scheme for end-to-end reservation-based network slicing in the backhaul and RAN. In both time-scales, the expected utilities of network tenants from serving users are maximized through two different mechanisms. In the long time-scale, each tenant decides whether or not to activate a slice to serve users, while in the short time-scale, the tenant reconfigures active slices to make them adaptive to demands of users and channel states. Resource management in both time-scales is implemented under two major assumptions: the user traffic and channel states vary over time and they are uncertain. 
	
	In the long time-scale, we design a slice utility function for each tenant based on the expected acquired revenue from users, the expected outage of downlinks, and the cost of slice activation. We formulate a sparse mixed-binary optimization problem to activate network slices if the expected utility of a tenant significantly improves after the activation. We use the $\ell_q$, $0<q<1$, regularization to tackle the non-convexity and non-continuity of the slice activation cost and also to promote binary solutions in the relaxed problem. We propose a Frank-Wolfe algorithm to successively minimize convex approximations of the original problem and jointly implement the slice configuration in the backhaul and RAN. Via numerical tests, we demonstrate that the proposed method obtains solutions that are near to the optimal ones. To the best of our knowledge, this is the first endeavor to mathematically study the sparse slice activation problem for unseen user traffic based on derived statistics.
	
	For network slicing in the short time-scale, we design a slice utility function for each tenant based on the acquired revenue from users, the QoS that the slice guarantees for its users, and the cost of slice reconfiguration. We formulate a sparse optimization problem to adaptively reconfigure network slices if the expected utility of a tenant significantly changes after the reconfiguration. We use the group Least Absolute Shrinkage and Selection Operator (LASSO) regularization to tackle the non-convexity and non-continuity of the slice reconfiguration cost. We propose an Alternating Direction Method of Multipliers (ADMM) algorithm to solve each (non-convex) group LASSO subproblem in the short time-scale. The proposed ADMM algorithm implements the slice reconfiguration in the backhaul through link capacity reservation via a fast, distributed algorithm that successively minimizes a convex approximation of the objective function and parallelizes computations across backhaul links. Furthermore, the proposed ADMM algorithm implements the slice reconfiguration in RAN through the transmission resource reservation for slices using 1) a proximal gradient descent method that decomposes the problem across slices; and 2) a bisection search method. We prove that the proposed ADMM algorithm converges to the global solution of each group LASSO subproblem despite its non-convexity. Extensive numerical simulations verify that the proposed approach outperforms the existing state of the art method.
	
	The rest of this paper is organized as follows. The system model is given in Section \ref{sec:model}. In Section \ref{sec:adap}, we formulate the optimization for slice activation in long time-scales and the optimization for slice reconfiguration in short time-scales. In Sections \ref{sec:slact} and \ref{sec:slrec}, we propose approaches to solve problems in long and short time-scales, respectively. The simulation results are given in Section \ref{sec:sim}, and concluding remarks are given in Section \ref{sec:con}.

	\section{System Model and Notations}\label{sec:model}
	Consider a typical scenario whereby user data is transmitted via backhaul network links from data centers to multiple geographically separated Access Points (APs) in RAN. Multiple APs jointly transmit the requested data to each user in a coordinated multi-point mode. We denote the set of mobile users by $\mathcal{K}=\{1,\dots,K\}$ and represent the set of APs in RAN  by $\mathcal{B}$. Furthermore, let $\mathcal{L}$ represent the set of backhaul links. The downlinks between APs and users are predetermined according to interference, path loss, shadowing and fading. A \textit{path} connects a data center and an AP through a sequence of wired links in the backhaul and then goes through one downlink to reach the end user. Since each user is served by multiple APs, several candidate paths are considered to connect APs to a data center. These paths share an identical origin (the data center) and destination (the end user). We denote a path by $p$ and represent the set of paths selected to carry user $k$ data by $\mathcal{P}_k$. We assume that a single commodity is requested by a user, and therefore, there are $K$ datastreams in the backhaul network. The proposed framework can be easily extended to a scenario in which each user demands multiple commodities.

	Suppose that a network tenant is denoted by $j$ where $\mathcal{J}=\{1,\dots,J\}$ is the set of all tenants.  Each tenant owns several slices, which differ in supported features and network function
	optimization. Multiple slices, which deliver similar features, can also be deployed by each tenant. However, they are
	responsible for serving different groups of users. A network slice is represented by $s$ and the set of slices possessed by tenant $j$ is denoted by $\mathcal{S}_j$. The set of users served by slice $s$ is denoted by $\mathcal{K}_s$ and the set of users served by tenant $j$ is represented by  $\mathcal{K}_j=\{\mathcal{K}_s\}_{s\in \mathcal{S}_j}$. The user-slice association is known and fixed. 
	
	In the considered model, each path belongs to one slice. The set of backhaul links on path $p$ (to serve user $k$) is denoted by $\mathcal{L}_k^p$. The set of network nodes on path $p$ is denoted by $\mathcal{U}_k^p$. The reserved rate for path $p$ to serve user $k$ is denoted by $r_k^p$. To wirelessly transmit the incoming data from each path, transmission resources should be sliced and reserved in APs. The two physical constraints that limit network resource slicing are as follows:
	\allowdisplaybreaks
	\begin{itemize}
		\item The aggregate amount of reserved traffic for those paths that go through a link cannot exceed the link capacity:
		\begin{align}
			\sum_{k=1}^K\sum_{p:\{p\in \mathcal{P}_k,l\in\mathcal{L}_k^p\}}r_k^p \leq C_l, \hspace{1cm}\forall l \in \mathcal{L},\label{eq:linkcap}
		\end{align}\normalsize
		where $C_l$ is the capacity of link $l$.
		\item The available resources in an AP are limited and are allocated to different downlinks created by that AP. The overall reserved transmission resources for those paths that share an AP must not exceed its capacity:
		\begin{align}
			\sum_{k=1}^K\sum_{p:\{p\in \mathcal{P}_k,b \in \mathcal{U}_k^p\}}t_k^p \leq C_b, \hspace{1cm}\forall b \in \mathcal{B},\label{eq:nodecap}
		\end{align}\normalsize
		where $t_k^p$ is the reserved transmission resource of AP $b$ to transmit incoming data from path $p$. Moreover, $C_b$ is the capacity of AP $b$.
		\item  The minimum aggregate reserved rate for users served by slice $s$ and tenant $j$ are denoted by $R_s^{\text{slc}}$ and $R_j^{\text{ten}}$, respectively, and we have
		\begin{align}
			\sum_{k \in\mathcal{K}_s}\sum_{p\in \mathcal{P}_k}r_k^p \geq R_s^{\text{slc}},\forall s,\:\:\text{and}\:\:\sum_{k \in\mathcal{K}_j}\sum_{p\in \mathcal{P}_k}r_k^p \geq R_j^{\text{ten}},\forall j.\label{eq:slicereq1}
		\end{align}\normalsize
		The constraints for minimum transmission resources for each slice and each tenant are as follows:
		\begin{align}
			\sum_{k \in\mathcal{K}_s}\sum_{p\in \mathcal{P}_k}t_k^p \geq B_s^{\text{slc}},\forall s,\:\:\text{and}\:\:\sum_{k \in\mathcal{K}_j}\sum_{p\in \mathcal{P}_k}t_k^p \geq B_j^{\text{ten}},\forall j.\label{eq:slicereq2}
		\end{align}
	\end{itemize}
	In addition to the above constraints, our multi-path model imposes another constraint. 
	\begin{itemize}
		\item The total reserved traffic for different paths that carry data to one user is equal to the reserved rate for that user. Hence, we have the following constraint:
		\begin{align}
			\sum_{p \in \mathcal{P}_k} r_k^p = r_k, \hspace{1cm}\forall k.\label{eq:ratesum}
		\end{align}
	\end{itemize}
	We denote the reserved link capacity for slice $s$ on link $l$ by $r_s^l$, which is calculated as follows:
	\begin{align}
		r_s^l=\sum_{k\in\mathcal{K}_s}\sum_{p:\{p\in \mathcal{P}_k,l\in\mathcal{L}_k^p\}}r_k^p, \hspace{1cm}\forall s, \forall l \in \mathcal{L}.\label{eq:linkres}
	\end{align}\normalsize
	Similarly, we denote the reserved transmission resources at AP $b$ by $t_s^b$, which is calculated as follows:
	\begin{align}
		t_s^b=\sum_{k\in\mathcal{K}_s}\sum_{p:\{p\in \mathcal{P}_k,b \in \mathcal{U}_k^p\}}t_k^p, \hspace{1cm}\forall s, \forall b \in \mathcal{B}.\label{eq:noderes}
	\end{align}\normalsize
	Furthermore, we define vectors of reserved resources for slices in the backhaul and RAN as $\mathbf{r}_s=\{r_s^l\}_{l\in\mathcal{L}}$ and $\mathbf{t}_s=\{t_s^b\}_{b \in \mathcal{B}}$.

	\section{Adaptive Multi-Tenant Network Slicing}\label{sec:adap}
	In this section, we propose a reservation-based network slicing approach. With the user demand and channel statistics, we adaptively optimize the activation and reconfiguration of network slices. Each tenant predicts future user demands and based on the expected revenue from serving users, it activates a number of its slices. After tenants activate a number of their slices to serve users, resource reservation for different slices is adaptively reconfigured across the network such that:
	\begin{enumerate}
		\item The expected revenue of tenants is maximized;
		\item The slice reconfiguration cost is minimized;
		\item The maximum isolation among reserved resources for slices is enforced; and
		\item The QoS requirements of users are met.
	\end{enumerate} 
	With a slotted time horizon, we consider a \textit{two time-scale} scheme to reserve resources for slices in the network. Resource reservations are carried out such that the revenue of tenants is maximized in long and short time-scales. Each tenant activates a subset of its slices to provide services to users in a long time-scale, while it reconfigures activated slices based on the statistics of user demands and channel capacities in short time-scales to improve the robustness of resource reservations and enhance QoS for users. The duration of each time-scale is chosen such that neither the statistics of user demands nor channel capacities changes within that time period. Examples for the duration of long and short time-scales are several days and a couple of hours, respectively \cite{pozza2020reconfiguring,lugones2017aidops}. We consider that the long and short time-scale durations are predetermined and kept fixed over time.

	\subsection{User Demand and Downlink Statistics}
	We assume that derived statistics remain identical in each time-scale. We add a subscript $S$ to the derived statistics for short time-scales  and $L$ to those derived for long time-scales. 
	The demand of user $k$, represented by $d_k$, follows a certain PDF, denoted by $f_{k,L_n}(d_k)$, in the $n^\text{th}$ long time-scale with a Cumulative Density Function (CDF) $F_{k,L_n}(d_k)$. Using the reserved rate for user $k$, $r_k$, the supportable demand of user $k$ is $\min(d_k,r_k)$ since the network can only provide $r_k$ to  user $k$ if the demand of user $k$ exceeds $r_k$. We leverage collected samples to estimate each PDF using Wolverton and Wagner estimator as discussed in \cite{comte2019bandwidth}. 

	In a practical fading environment, the transmission rate to each user in the coverage area depends on the random channel capacity (i.e., instantaneous achievable rate), which is a function of the amount of resources, e.g., bandwidth, supplied to the downlink \cite{1210731,musavian2015effective,4411539}. In the proposed model, we do not make any assumption about the type of allocated transmission resources in APs. This can be power, bandwidth, or time-slot. Since a path connects a data center to a user, a given path uniquely identifies the downlink by which a user is served. The achievable rate of a downlink, denoted by $v_k^p$, can follow any arbitrary PDF. Let $z_{k,L_n}^p(v_k^p,t_k^p)$ denote the PDF of the achievable rate of the downlink of path $p$ in the $n^\text{th}$ long time-scale, where the amount of the transmission resource supplied to the downlink is $t_k^p$. The CDF of the downlink achievable rate is denoted by $Z_{k,L_n}^p(v_k^p,t_k^p)$ in the $n^{\text{th}}$ long time-scale and determines the probability that the achievable rate of the downlink of path $p$ to serve user $k$ is at most $v_k^p$.  When the random achievable rate of a downlink is less than the reserved rate  $r_k^p$, the experienced outage is $r_k^p-v_k^p$. The probability that this amount of outage occurs in the $n^{\text{th}}$ short time-scale is $z_{k,S_n}^p(v_k^p,t_k^p)$.  In  light of the above arguments, the expected value of the outage of the downlink of path $p$ is obtained as
	\allowdisplaybreaks
	\begin{align}
		\int_{0}^{r_k^p}z_{k,S_n}^p(v_k^p,t_k^p)\:(r_k^p-v_k^p)dv_k^p,\hspace{.5cm}\forall p\in \mathcal{P}_k,\forall k.\label{eq:out1}
	\end{align}
	Since the achievable rate is a continuous random variable, we have the above integral. As $z_{k,S_n}^p(v_k^p,t_k^p)$ is always non-negative, then the expected outage is non-decreasing in $r_k^p$. In addition, we assume that the expected outage is non-increasing in $t_k^p$.

	Consider that $\phi_{k,S_n}(\cdot)$ is the revenue function of a tenant from serving user $k$ in the $n^{\text{th}}$ short time-scale. We consider $\phi_{k,S_n}(\cdot)$ to be a concave and non-decreasing function, e.g., $\phi_{k,S_n}(x)=1-\exp(-x), x\geq 0$. The expected revenue gained from serving user $k$ is calculated as follows:
	\begin{align}
		&\mathbb{E}_{d_k}\Bigg[\phi_{k,S_n}\Big(\min( r_k,d_k)\Big)\Bigg]=\int_0^{ r_k}\hspace{-.2cm}\phi_{k,S_n}(y)f_{k,S_n}(y)dy\nonumber\\
		&+\int_{ r_k}^\infty\hspace{-.2cm} \phi_{k,S_n}(r_k)f_{k,S_n}(y)dy.\label{eq:util1}
	\end{align}
	In the first integral, the random demand lies in $\interval[{0,\:r_k}]$, and in the second $d_k\in\interval[{r_k,\infty})$. Using this revenue function, we can maximize the expected supportable rates of users through the maximization of the revenue functions.
	
	Consider that $\phi_{k,L_n}(\cdot)$ is the revenue function of a tenant acquired by serving user $k$ in long time-scales. We consider $\phi_{k,L_n}(\cdot)$ to be a concave and non-decreasing function. In long time-scales, the set of users served by each slice, denoted by $\mathcal{K}_s$, is random and follows a certain probability mass function. To find the expected revenue of each slice from serving users in a long time-scale, we need to also take an expectation with respect to $\mathcal{K}_s$. The expected gained revenue by slice $s$ is the summation of expected revenue obtained from users served by slice $s$ \cite{caballero2018network} and is calculated as follows:
	\begin{align}
		&\mathbb{E}_{\mathcal{K}_s}\Bigg[\sum_{k\in\mathcal{K}_s}\mathbb{E}_{d_k}\Big[\phi_{k,L_n}\Big(\min( r_k,d_k)\Big)\Big]\Bigg].\label{eq:util2}
	\end{align}\normalsize
	Since the set of users served by each slice is uncertain in long time-scales, we need to take an expectation with respect to $\mathcal{K}_s$ from \eqref{eq:linkcap}--\eqref{eq:slicereq2} and \eqref{eq:out1}. 
	The argument of revenue functions $\phi_{k,S_n}(\cdot)$ and $\phi_{k,L_n}(\cdot)$ in \eqref{eq:util1} and \eqref{eq:util2}, which is $\min( r_k,d_k)$, is a concave function of $r_k$. Since $\phi_{k,S_n}(\cdot)$ and $\phi_{k,L_n}(\cdot)$ are both non-decreasing concave functions, based on the rule for function compositions in \cite[eq. (3.10)]{boyd2004convex}, revenue functions are also concave. After taking expectations, \eqref{eq:util1} and \eqref{eq:util2} remain concave functions.

	\subsection{Slice Activation}
	A tenant operates a slice only if the acquired revenue from activating that slice is considerable. The slice activation involves binary variables, which determine whether or not a slice is activated.  Suppose that $x_{s,1}$ is a binary variable and if $x_{s,1}=1$, then slice $s$ is activated and $x_{s,1}=0$ if slice $s$ is not activated. In addition, we consider the binary $x_{s,2}$ as a complement variable for $x_{s,1}$ such that
	\begin{align}
		x_{s,1}+x_{s,2}=1, \hspace{.5cm}\hspace{.5cm}\forall s\in \mathcal{S}_j,\forall j.\label{eq:sum}
	\end{align}
	The cost of activating a slice is denoted by $c_a$ and the cost of activated slices of tenant $j$ is \begin{align}
		c_a\norm{\{x_{s,1}\}_{s\in \mathcal{S}_j}}_0.\label{eq:cost}
	\end{align} 
	When $x_{s,1}=1$, this non-zero element is counted by the $\ell_0$-norm. While each tenant desires to maximize the expected gained revenue from users by increasing reserved rates, to improve the user QoS, it minimizes the expected outage of downlinks. Therefore, the overall utility function of the $j^{\text{th}}$ tenant from the activation of a subset of slices is
	\begin{align}
		&\sum_{s \in\mathcal{S}_j}\hspace{-0.1cm}\mathbb{E}_{\mathcal{K}_s}\hspace{-0.1cm}\Bigg[\sum_{k\in\mathcal{K}_s}\hspace{-0.1cm}\mathbb{E}_{d_k}\Big[\phi_{k,L_n}\Big(\hspace{-0.1cm}\min( r_k,d_k)\Big)\Big]\Bigg]-c_a\norm{\{x_{s,1}\}_{s\in \mathcal{S}_j}}_0\nonumber\\
		&-\sum_{s \in\mathcal{S}_j}\theta_s\mathbb{E}_{\mathcal{K}_s}\sum_{k\in\mathcal{K}_s}\left[\sum_{p\in \mathcal{P}_{k}}\int_{0}^{r_k^p}z_{k,L_n}^p(v_k^p,t_k^p)\:(r_k^p-v_k^p)dv_k^p\right],\label{eq:util}
	\end{align}
	where $\theta_s$ is a constant adjusted by the system engineer.
	
	When a slice is activated, resources are reserved for that slice. To relate $x_{s,1}$ to the reserved resources for a slice, we consider the following constraints:
	\begin{align}
		& \mathbb{E}_{\mathcal{K}_s}\left[\sum_{k\in\mathcal{K}_s}\sum_{p\in \mathcal{P}_k}r_k^p\right]\geq x_{s,1}R_s^{\text{slc}},\forall s\in \mathcal{S}_j,\forall j,\label{eq:rateslice}\\
		& \mathbb{E}_{\mathcal{K}_s}\left[\sum_{k\in\mathcal{K}_s}\sum_{p\in \mathcal{P}_k}t_k^p\right]\geq x_{s,1}B_s^{\text{slc}},\forall s\in \mathcal{S}_j,\forall j.\label{eq:exbandslice}
	\end{align}
	We note that \eqref{eq:rateslice} and \eqref{eq:exbandslice} imply that the slice requirement constraints are only considered for each activated slice $s$. Additionally, we consider
	\begin{align}
		& \mathbb{E}_{\mathcal{K}_s}\left[\sum_{k\in\mathcal{K}_s}\sum_{p\in \mathcal{P}_k}r_k^p\right]\leq x_{s,1}\Psi,\forall s\in \mathcal{S}_j,\forall j,\label{eq:actrate}\\
		& \mathbb{E}_{\mathcal{K}_s}\left[\sum_{k\in\mathcal{K}_s}\sum_{p\in \mathcal{P}_k}t_k^p\right]\leq x_{s,1}\Psi,\forall s\in \mathcal{S}_j,\forall j,\label{eq:actband}
	\end{align}
	to enforce zero resource reservation for a slice when $x_{s,1}=0$.  In \eqref{eq:actrate} and \eqref{eq:actband} , $\Psi$ is a large number and when $x_{s,1}=1$, \eqref{eq:actrate} and \eqref{eq:actband} are relaxed. The sparse slice activation optimization is formulated as follows:
	\begin{subequations}\label{opt:activate}
		\begin{align}
			\allowdisplaybreaks
			& \underset{\mathbf{r,t,x}}{\text{min}}
			& & -\sum_{j=1}^{J}\psi_j\sum_{s \in\mathcal{S}_j}\mathbb{E}_{\mathcal{K}_s}\Bigg[\sum_{k\in\mathcal{K}_s}\mathbb{E}_{d_k}\Big[\phi_{k,L_n}\Big(\min( r_k,d_k)\Big)\Big]\Bigg]\hspace{-.0cm}\nonumber\\
			&&&\hspace{-0.9cm}+\sum_{j=1}^{J}\psi_j\sum_{s \in\mathcal{S}_j}\theta_s\mathbb{E}_{\mathcal{K}_s}\sum_{k\in\mathcal{K}_s}\left[\sum_{p\in \mathcal{P}_{k}}\int_{0}^{r_k^p}\hspace{-.3cm}z_{k,L_n}^p(v_k^p,t_k^p)\:(r_k^p-v_k^p)dv_k^p\right]\nonumber\\
			&&&\hspace{-0.9cm}+\sum_{j=1}^{J}c_a\norm{\{x_{s,1}\}_{s\in \mathcal{S}_j}}_0\label{opt:act}\\
			& \text{s.t.}&& \sum_{j=1}^{J}\sum_{s \in\mathcal{S}_j}\mathbb{E}_{\mathcal{K}_s}\left[\sum_{k\in\mathcal{K}_s}\sum_{p:\{p\in \mathcal{P}_k,l\in\mathcal{L}_k^p\}}r_k^p\right]\leq C_l, \forall l \label{eq:explink}\\
			&&& \sum_{j=1}^{J}\sum_{s \in\mathcal{S}_j}\mathbb{E}_{\mathcal{K}_s}\left[\sum_{k\in \mathcal{K}_s}\sum_{p:\{p\in \mathcal{P}_k,b \in \mathcal{U}_k^p\}}t_k^p\right]\leq C_b,\forall b\label{eq:expnode}\\
			&&& \sum_{s \in\mathcal{S}_j}\mathbb{E}_{\mathcal{K}_s}\left[\sum_{k\in\mathcal{K}_s}\sum_{p\in \mathcal{P}_k}r_k^p\right]\geq R_j^{\text{ten}}\label{eq:slcrateten},\forall j,\\
			&&& \sum_{s \in\mathcal{S}_j}\mathbb{E}_{\mathcal{K}_s}\left[\sum_{k\in\mathcal{K}_s}\sum_{p\in \mathcal{P}_k}t_k^p\right]\geq B_j^{\text{ten}},\forall j,\label{eq:expbandten}\\
			&&& \eqref{eq:ratesum},\eqref{eq:sum},\eqref{eq:rateslice},\eqref{eq:exbandslice}, \eqref{eq:actrate}, \eqref{eq:actband},\nonumber\\
			&&& x_{s,1},x_{s,2}\in\{0,1\}, r_k^p, t_k^p\geq 0, \forall p\in \mathcal{P}_k,\forall k,\forall s,\nonumber
		\end{align}
	\end{subequations}
	where $\psi_j$ is a positive weight given to tenant  $j$ in order to adjust priorities. Moreover, $\mathbf{x}=\{x_{s,1},x_{s,2}\}_{s\in \mathcal{S}_j,j\in\mathcal{J}}$. 
	
	\begin{remark}
		Suppose that multiple paths available to user $k$ share a downlink (the last hop). The aggregate outage of downlinks for serving user $k$ is calculated as follows:
		\begin{align}
			\sum_{w \in \mathcal{W}_k}\int_{0}^{\sum_{p:\{p\in \mathcal{P}_k,w\in p\}}r_k^p}\hspace{-.2cm} z_{k,L_n}^w(v_{k}^w,t_{k}^w)(\hspace{-.2cm}\sum_{p:\{p\in \mathcal{P}_k,w\in p\}}\hspace{-.5cm}r_k^p-v_{k}^w)dv_{k}^w\label{eq:mul},
		\end{align}
		where $\mathcal{W}_k$ is the set of downlinks, each denoted by $w$, for serving user $k$. When multiple paths available to user $k$ share a downlink, the above outage is placed in \eqref{opt:act} instead of its second term, which includes \eqref{eq:out1}.
	\end{remark}

	\subsection{Slice Reconfiguration}
	Here, we consider slicing network resources, both in the backhaul and RAN, among several network tenants in short time-scales. The number of users that are served by each slice is known in short time-scales. As PDFs of user demands and achievable rates of downlinks change over time, allocated resources to each slice should be adapted. However, this change involves a reconfiguration cost and service interruptions. Therefore, an adaptive approach should reconfigure each slice only if the utility function for that tenant significantly changes with new PDFs. The state of a slice is reflected by the reserved resources in the backhaul and RAN. The reconfiguration cost of a slice is a discrete function of the state difference \cite{wang2019reconfiguration}. We underline the vector of reserved resources in each short time-scale. If we have $\norm{\underline{\mathbf{r}}_s^{n}-\underline{\mathbf{r}}_s^{n-1}}_2 >0$ or $\norm{\underline{\mathbf{t}}_s^{n}-\underline{\mathbf{t}}_s^{n-1}}_2  >0$, then slice $s$ varies in two consecutive time-scales. The reconfiguration cost is proportional to the number of layout changes for slices. We use $\ell_0$-norm to detect reconfigurations for slices. Consider that the slice reconfiguration cost is denoted by $c_r$. To minimize the slice reconfiguration cost in the network from time $n-1$ to $n$, we minimize \allowdisplaybreaks
	\begin{align}
		&c_r\big(\norm{\{\norm{\underline{\mathbf{r}}_s^{n}-\underline{\mathbf{r}}_s^{n-1}}_2\}_{s \in\mathcal{S}_j,j\in\mathcal{J}}}_0\hspace{-0.2cm}+\hspace{-0.1cm}\norm{\{\norm{\underline{\mathbf{t}}_s^{n}-\underline{\mathbf{t}}_s^{n-1}}_2\}_{s \in\mathcal{S}_j,j\in\mathcal{J}}}_0\big).\label{eq:change}
	\end{align}\normalsize
	
	Here, we discuss the isolation of reserved resources for slices. When the random traffic demand from slice $s$ exceeds the reserved traffic rate for it, the tenant needs to allocate more resources to support the demand. This can violate the resource  reservation for other slices and hurt the QoS for users served by other slices. Each slice desires to isolate its reserved resources as much as possible \cite{chien2019end,yan2019intelligent}. Although static network slicing provides complete resource isolation among slices, it performs poorly due to its inflexibility in supporting time-varying user demands \cite{lee2018efficient}. The dynamic network slicing better supports user demands, although it increases the risk for violation of resource isolation for slices. To increase the slice isolation, we minimize the expected excessive demand from each slice, which is calculated as follows: 
	\begin{align}
		\int_{\sum_{k \in\mathcal{K}_s}\sum_{p\in \mathcal{P}_{k}}r_{k}^{p}}^{\infty}(y-\sum_{k \in\mathcal{K}_s}\sum_{p\in \mathcal{P}_{k}}r_{k}^{p})f_{s,S_n}(y)dy.\label{eq:iso}
	\end{align}
	In the above expression, $y$ is the integration variable and corresponds to the aggregate demanded rate from slice $s$, $\sum_{k \in\mathcal{K}_s}d_k$, and when $y\in \interval[{\sum_{k \in\mathcal{K}_s}\sum_{p\in \mathcal{P}_{k}}r_{k}^{p},\infty})$, the slice is unable to support the demand. In \eqref{eq:iso}, $f_{s,S_n}(\cdot)$ is the PDF for the aggregate demand from slice $s$. We utilize the data-driven density estimator given in \cite{comte2019bandwidth} to numerically estimate $f_{s,S_n}(\cdot)$ for each short time-scale.

	Each slice provides its users with particular QoS guarantees. We consider that each slice guarantees that the expected outage \eqref{eq:out1} for a served user is less than a certain fraction, denoted by $\beta_s$, of the reserved rate for that user in a short time-scale. Therefore, we have
	\begin{align}
		\int_{0}^{r_k^p}z_{k,S_n}^p(v_k^p,t_k^{p})\:(r_k^{p}-v_k^p)dv_k^p\leq \beta_sr_k^{p},\hspace{1cm}\forall k\in \mathcal{K}_s.\label{eq:out2}
	\end{align}
	The objective function that we minimize is 
	\begin{align}\textstyle
		&\Gamma(\mathbf{r})= \sum_{j=1}^{J}\sum_{s \in\mathcal{S}_j}\int_{\sum_{k \in\mathcal{K}_s}\sum_{p\in \mathcal{P}_{k}}r_{k}^{p}}^{\infty}\hspace{-.4cm}(y-\sum_{k \in\mathcal{K}_s}\sum_{p\in \mathcal{P}_{k}}r_{k}^{p})f_{s,S_n}(y)dy\nonumber\\
		& -\sum_{j=1}^{J}\psi_j\Bigg(\sum_{s \in\mathcal{S}_j}\sum_{k \in\mathcal{K}_s}\mathbb{E}_{d_k}\Bigg[\phi_{k,S_n}\Big(\min(\sum_{p \in \mathcal{P}_k} r_k^{p},d_k)\Big)\Bigg]\Bigg).\label{eq:obj}
	\end{align}
	Using \eqref{eq:obj}, the formulated joint optimization to minimize the reconfiguration costs, maximize the isolation of reserved resources, guarantee QoS for slice users, and maximize the utilities of tenants from serving their users in short time-scales is as follows:
	\allowdisplaybreaks
	\begin{align}
		& \underset{\mathbf{r},\mathbf{t}\geq \mathbf{0}}{\text{min}}
		& & \hspace{-.2cm}\Gamma(\mathbf{r})+c_r\norm{\{\norm{\mathbf{r}_s-\underline{\mathbf{r}}_s^{n-1}}_2\}_{s \in\mathcal{S}_j,j\in\mathcal{J}}}_0\nonumber\\
		&&&+c_r\norm{\{\norm{\mathbf{t}_s-\underline{\mathbf{t}}_s^{n-1}}_2\}_{s \in\mathcal{S}_j,j\in\mathcal{J}}}_0\label{opt:slice1}\\
		& \text{s.t.}
		& & \eqref{eq:linkcap}, \eqref{eq:nodecap}, \eqref{eq:slicereq1}, \eqref{eq:slicereq2}, \eqref{eq:linkres}, \eqref{eq:noderes}, \eqref{eq:out2}.\nonumber
	\end{align}
	\begin{remark}
		If multiple paths available to user $k$ share a downlink, \eqref{eq:out2} can be rewritten as follows:
		\begin{align}
			&\int_{0}^{\sum_{p:\{p\in \mathcal{P}_k,w\in p\}}r_k^p}\hspace{-.2cm} z_{k,S_n}^w(v_{k}^w,t_{k}^w)(\hspace{-.2cm}\sum_{p:\{p\in \mathcal{P}_k,w\in p\}}\hspace{-.5cm}r_k^p-v_{k}^w)dv_{k}^w\nonumber\\
			&\leq \beta_s(\sum_{p:\{p\in \mathcal{P}_k,w\in p\}}r_k^p),\hspace{1cm}\forall k\in \mathcal{K}_s. \label{eq:rewritte}
		\end{align}.
	\end{remark}

	\section{The Proposed Approach for Slice Activation}\label{sec:slact}
	In optimization \eqref{opt:activate}, the constraints \eqref{eq:rateslice}--\eqref{eq:actband} and \eqref{eq:explink}--\eqref{eq:expbandten} can be rewritten in the affine form after we substitute a weighted average for the expectation with respect to the set of users associated with a slice, i.e., $\mathbb{E}_{\mathcal{K}_s}[\cdot]$. For example, we can rewrite \eqref{eq:explink} in the following form:
	\begin{align}
		\sum_{j=1}^{J}\sum_{s \in\mathcal{S}_j}\sum_{\mathcal{K}_s^u}P_{u}(\sum_{k\in\mathcal{K}_s^u}\sum_{p:\{p\in \mathcal{P}_k,l\in\mathcal{L}_k^p\}}r_k^{p})\leq C_l,\hspace{.3cm}l\in\mathcal{L},\nonumber
	\end{align}
	where $\mathcal{K}_s^u$ is the $u^{\text{th}}$ possible set of potential users served by slice $s$ with a $P_u$ probability of occurring. One can similarly rewrite \eqref{eq:rateslice}--\eqref{eq:actband} and \eqref{eq:expnode}--\eqref{eq:expbandten} in the above affine form.
	
	We let $\mathbf{x}_{j,1}=\{x_{s,1}\}_{s\in \mathcal{S}_j}$ and $\mathbf{x}_{j,2}=\{x_{s,2}\}_{s\in \mathcal{S}_j}$.  Before developing an approach to solve the slice activation optimization, we consider  the following $\ell_q$-norm:
	
	\begin{align}\norm{\mathbf{x}_{j,1}+\epsilon\mathbf{1}}_q^q,\label{opt:ub}
	\end{align}
	where $0<q<1$ and $\epsilon$ is a small positive number to make \eqref{opt:ub} differentiable. When $q$ and $\epsilon$ are close to zero, one can use $c_a\norm{\mathbf{x}_{j,1}+\epsilon\mathbf{1}}_q^q$ to efficiently approximate \eqref{eq:cost} \cite{4303060,7101837,lai2011unconstrained}. 
	
	Although \eqref{opt:ub} approximates the $\ell_0$-norm, we note that $\norm{\mathbf{x}_{j,1}+\epsilon\mathbf{1}}_q^q$ is concave and \eqref{opt:ub} is not easy to minimize. Therefore, if we add $c_a\norm{\mathbf{x}_{j,1}+\epsilon\mathbf{1}}_q^q$ to the objective function of \eqref{opt:activate}, it is still hard to solve. To tackle this problem, we consider a quadratic upper-bound \cite[eq. (12)]{hong2015unified} for \eqref{opt:ub} and successively minimize the upper-bound  \cite{razaviyayn2013unified}. We denote the $i^{\text{th}}$ iterate of $\mathbf{x}_{j,1}$ by $\underline{\mathbf{x}}_{j,1}^i$ and obtain the following upper-bound:
	\begin{align}
		&\norm{\mathbf{x}_{j,1}+\epsilon\mathbf{1}}_q^q \leq \norm{\underline{\mathbf{x}}_{j,1}^i+\epsilon\mathbf{1}}_q^q\nonumber\\
		&+q\big\langle\{(\underline{x}_{s,1}^i+\epsilon)^{q-1}\}_{s\in \mathcal{S}_j}, (\mathbf{x}_{j,1}-\underline{\mathbf{x}}_{j,1}^i)\big\rangle+\frac{e}{2}\norm{\mathbf{x}_{j,1}-\underline{\mathbf{x}}_{j,1}^i}_2^2,\label{eq:up}
	\end{align}
	where $e>0$ is a small number and $\langle\cdot,\cdot\rangle$ represents the inner product.
	We also relax the constraint $x_{s,1},x_{s,2}\in\{0,1\}$ and instead include $0\leq x_{s,1} \leq 1$ and $0\leq x_{s,2} \leq 1$. We substitute the RHS of \eqref{eq:up} for $\norm{\{x_{s,1}\}_{s\in \mathcal{S}_j}}_0$ and iteratively solve \eqref{opt:activate}.
	The above upper-bound approximation is locally tight up to the first order. In other words, in any arbitrary point $\mathbf{x}_{j,1}=\underline{\mathbf{x}}_{j,1}^i$, $\norm{\mathbf{x}_{j,1}+\epsilon\mathbf{1}}_q^q$ and the upper-bound have the same value and the same gradient. The above upper-bound is also continuous. Thus, the upper-bound satisfies all four convergence conditions given in \cite[Assumption 2]{razaviyayn2013unified}. Based on \cite[Theorem 2]{razaviyayn2013unified}, the obtained solution by the successive upper-bound minimization  is a stationary (KKT) solution to problem \eqref{opt:activate} with relaxed $\mathbf{x}$ and \eqref{opt:ub} in the objective function.
	
	We note that \eqref{opt:activate} is convex with respect to $\mathbf{r}$. However, \eqref{opt:act} is not always convex in $\mathbf{t}$. To address this issue, we add the following proximal terms \begin{align}
		\sum_{j=1}^{J}\sum_{k\in\mathcal{K}_j}\sum_{p\in \mathcal{P}_{k}}\frac{\zeta_{k}^p}{2}\norm{t_{k}^p-\hat{t}_{k}^p}_2^2,\label{eq:ubkir}
	\end{align} \normalsize
	with sufficient weights $\zeta_{k}^p>0$ to the objective function of \eqref{opt:activate} with the purpose of convexifying the objective function with respect to $\mathbf{t}$ locally. In the proximal term, $\hat{t}_{k}^{p}$ is the most recent iterate.  The objective function with proximal terms is an upper-bound for the original objective function, and we successively minimize the upper-bound \cite{razaviyayn2013unified}.

	To efficiently solve \eqref{opt:activate} in each iteration of minimizing \eqref{opt:act} with \eqref{eq:up} and \eqref{eq:ubkir}, we propose a Frank-Wolfe algorithm. The Frank-Wolfe algorithm linearizes the objective function and finds a feasible descent direction within the set of constraints \cite[Sec. 3.2.2]{bertsekas1999nonlinear}.  We superscript the $m^\text{th}$ iterate of the Frank-Wolfe algorithm by $m$. Using the linearized objective function, the direction-finding subproblem in the $m+1^\text{th}$ iteration of the Frank-Wolfe algorithm becomes
	\begin{align}
		\allowdisplaybreaks
		& \underset{\overline{\mathbf{r}},\overline{\mathbf{t}},\overline{\mathbf{x}}\geq \mathbf{0}}{\text{min}}
		& &	\sum_{j=1}^{J}\psi_j\sum_{s \in\mathcal{S}_j}\sum_{\mathcal{K}_s^u}P_{u}\Big(\sum_{k\in\mathcal{K}_s^u}\sum_{p\in\mathcal{P}_k}\hspace{-.0cm}\omega_k^p(r_k^{p,m},t_k^{p,m})\overline{t}_k^p\Big)\nonumber\\
		&&&\hspace{-.9cm}+\sum_{j=1}^{J}\psi_j\sum_{s \in\mathcal{S}_j}\sum_{\mathcal{K}_s^u}P_{u}\Big(\sum_{k\in\mathcal{K}_s^u}\sum_{p\in \mathcal{P}_k}\hspace{-.0cm}\chi_k^p(\{r_k^{p,m}\}_{p\in \mathcal{P}_k},t_k^{p,m})\overline{r}_k^p\Big)\hspace{-.0cm}\nonumber\\
		&&&\hspace{-.9cm}+\sum_{j=1}^{J}\sum_{s \in\mathcal{S}_j}\varrho_{s,1}^1(x_{s,1}^{m})\overline{x}_{s,1}\nonumber\\
		& \text{s.t.}&&\hspace{-.5cm} \eqref{eq:sum},\eqref{eq:rateslice}-\eqref{eq:actband},\eqref{eq:explink}-\eqref{eq:expbandten},0\leq \overline{x}_{s,1},\overline{x}_{s,2}\leq 1,\label{opt:t}
	\end{align}
	where $\omega_k^p(r_k^{p,m},t_k^{p,m})$ and $\chi_k^p(\{r_k^{p,m}\}_{p\in \mathcal{P}_k},t_k^{p,m})$ are 
	\begin{align}\textstyle
		&\omega_k^p(r_k^{p,m},t_k^{p,m})=\zeta_k^p(t_{k}^{p,m}-\hat{t}_{k}^{p})\nonumber\\
		&+\theta_s\int_{ 0}^{r_k^{p,m}}\Big(r_k^{p,m}-v_k^p\Big)\frac{\partial z_{k,L_n}^p(v_k^p,t_k^{p})}{\partial t_k^{p} }\mid_{t_k^{p}=t_k^{p,m}}dv_k^p,\label{eq:lin1}
	\end{align}
	and 
	\begin{align}\textstyle
		&\chi_k^p(\{r_k^{p,m}\}_{p\in \mathcal{P}_k},t_k^{p,m})=\theta_sZ_{k,L_n}^p(r_k^{p,m},t_{k}^{p,m})\nonumber\\
		&-\Bigg(\phi_{k,L_n}'\Big(\sum_{p \in \mathcal{P}_k} r_k^{p,m}\Big)\Big(1-F_{k,L_n}(\sum_{p \in \mathcal{P}_k} r_k^{p,m})\Big)\Bigg).\label{eq:lin}
	\end{align}
	Furthermore, $\varrho_{s,1}^1(x_{s,1}^{m})$ is calculated as follows:
	\begin{align}
		&\varrho_{s,1}^1(x_{s,1}^{m})=c_a\left(\frac{q}{(\underline {x}_{s,1}^{i}+\epsilon)^{1-q}}+e(x_{s,1}^{m}-\underline{x}_{s,1}^i)\right).\label{eq:xdir1}
	\end{align}
	The solution to \eqref{opt:t} can be easily obtained using CPLEX or Gurobi.	
	We update variables in the $m+1^{\text{th}}$ iteration of the Frank-Wolfe algorithm as follows:
	\begin{subequations}
		\begin{eqnarray}
			&t_k^{p,m+1}=t_k^{p,m}+\pi^m(\overline{t}_k^{p}-t_k^{p,m}),\label{eq:update1}\\
			&r_k^{p,m+1}=r_k^{p,m}+\pi^m(\overline{r}_k^p-r_k^{p,m}),\label{eq:update2}\\
			&x_{s,1}^{m+1}=x_{s,1}^{m}+\pi^m(\overline{x}_{s,1}-x_{s,1}^{m}),\label{eq:update3}\\
			&x_{s,2}^{m+1}=x_{s,2}^{m}+\pi^m(\overline{x}_{s,2}-x_{s,2}^{m}),\label{eq:update4}
		\end{eqnarray}
	\end{subequations}
	where $\pi^m=\frac{2}{2+m}$. We update $\omega_k^p(r_k^{p,m+1},t_k^{p,m+1})$, $\chi_k^p(\{r_k^{p,m+1}\}_{p\in \mathcal{P}_k},t_k^{p,m+1})$, and $\varrho_{s,1}^1(x_{s,1}^{m+1})$ as given in \eqref{eq:lin1}, \eqref{eq:lin}, and \eqref{eq:xdir1}, respectively. We continue until each $r_k^{p,m}$, $t_k^{p,m}$, $x_{s,1}^{m}$, and $x_{s,2}^{m}$ converges. After the convergence of the Frank-Wolfe algorithm, we update $\hat{\mathbf{t}}=\mathbf{t}^{m}$. We continue solving with the Frank-Wolfe algorithm until $\hat{\mathbf{t}}$ converges. The summary of the proposed Frank-Wolfe approach is given in Algorithm \ref{al:parallel1}.
	\begin{remark}
		If \eqref{eq:mul} is considered in \eqref{opt:act}, then \eqref{eq:lin1} and \eqref{eq:lin} are changed to 
		\begin{align}\textstyle
			&\omega_k^w(\{r_k^{p,m}\}_{p:\{p\in \mathcal{P}_k,w\in p\}},t_k^{p,w})\nonumber\\
			&=\theta_s\int_{ 0}^{\sum_{p:\{p\in \mathcal{P}_k,w\in p\}}r_k^{p,m}}\hspace{-.2cm}\Big(\hspace{-.4cm}\sum_{p:\{p\in \mathcal{P}_k,w\in p\}}\hspace{-.7cm}r_k^{p,m}-v_k^w\Big)\nonumber\\
			&\times\frac{\partial z_{k,L_n}^p(v_k^w,t_k^{w})}{\partial t_k^{w} }\mid_{t_k^{w}=t_k^{w,m}}dv_k^w+\zeta_k^w(t_{k}^{w,m}-\hat{t}_{k}^{w}),  w\in p, \forall k,\nonumber
		\end{align}
		and 
		\begin{align}\textstyle
			&\chi_k^p(\{r_k^{p,m}\}_{p\in \mathcal{P}_k},t_k^{w,m})=\theta_sZ_{k,L_n}^p(\sum_{p:\{p\in \mathcal{P}_k,w\in p\}}r_k^{p,m},t_{k}^{w,m})\nonumber\\
			&-\Bigg(\phi_{k,L_n}'\Big(\sum_{p \in \mathcal{P}_k} r_k^{p,m}\Big)\Big(1-F_{k,L_n}(\sum_{p \in \mathcal{P}_k} r_k^{p,m})\Big)\Bigg),  w\in p, \forall k.\nonumber
		\end{align}
	\end{remark}
	\begin{algorithm}[t!]
		0.	\textbf{Initialization} Assign values to $\mathbf{r}^{0}$, $\mathbf{t}^{0}$, $\mathbf{x}_{j,1}^{0}$, $\mathbf{x}_{j,2}^{0}$, $m=0$; \\
		\Repeat{$\hat{\mathbf{t}}$ converges}{
			1. $\hat{\mathbf{t}}=\mathbf{t}^{m}$\;
			\Repeat{$\mathbf{r}^m,\mathbf{t}^m,\mathbf{x}^m$ converges}{
				2. Find $\omega_k^p(r_k^{p,m},t_k^{p,m})$, $\chi_k^p(\{r_k^{p,m}\}_{p\in \mathcal{P}_k},t_k^{p,m})$, and  $\varrho_{s,1}^1(x_{s,1}^{m})$\;
				3. Solve \eqref{opt:t} and update variables using \eqref{eq:update1}--\eqref{eq:update4}\;
				4. $m=m+1$\;
			}
		}     
		5. $\underline{\mathbf{x}}^{i+1}=\mathbf{x}^{m}$;
		\caption{The proposed Frank-Wolfe Algorithm}
		\label{al:parallel1}
	\end{algorithm}
	
	After Algorithm \ref{al:parallel1} converges, we update $\underline{\mathbf{x}}^{i+1}$ as $\underline{\mathbf{x}}^{i+1}=\mathbf{x}^{m}$ and solve again with Algorithm \ref{al:parallel1}. When we solve a few iterations with Algorithm \ref{al:parallel1}, it is possible that $\underline{\mathbf{x}}^{i}$ does not become binary. We use $\ell_q$-regularization to promote binary solutions. Consider the following optimization problem:
	\begin{subequations}\label{opt:sp}
		\begin{align}
			& \underset{\textbf{x}_s}{\text{min}}
			& & \hspace{-1.5cm}\norm{\mathbf{x}_s+\epsilon\mathbf{1}}_q^q\label{opt:sp1}\\
			& \text{s.t.}
			& &\hspace{-1.5cm}\eqref{eq:sum}, 0\leq x_{s,1} \leq 1,0\leq x_{s,2}\leq 1,\forall s,\label{opt:sp3}
		\end{align}\normalsize
	\end{subequations}
	where $\mathbf{x}_s=\{x_{s,1},x_{s,2}\}$. The optimal solution of \eqref{opt:sp} is always binary, i.e., $x_{s,1}^*,x_{s,2}^*\in\{0,1\}$ \cite{liu2015iterative,jiang2016l_p}. Similar to \eqref{opt:ub}, \eqref{opt:sp1} is concave. Therefore, we consider a quadratic upper-bound \cite[eq. (12)]{hong2015unified} for it as follows:
	\begin{align}
		&\norm{\mathbf{x}_s+\epsilon\mathbf{1}}_q^q\leq \norm{\underline{\mathbf{x}}_s^i+\epsilon\mathbf{1}}_q^q+q\big\langle\left\{(\underline{x}_{s,1}^{i}+\epsilon)^{q-1},(\underline{x}_{s,2}^{i}+\epsilon)^{q-1}\right\}\nonumber\\
		&, (\mathbf{x}_s-\underline{\mathbf{x}}_s^i)\big\rangle+\frac{e}{2}\norm{\mathbf{x}_s-\underline{\mathbf{x}}_s^i}_2^2,\hspace{.3cm}\forall s.\label{eq:ub4}
	\end{align} 
	We can give a weight $\gamma^i$ to the RHS of \eqref{eq:ub4} and include it for all slices in \eqref{opt:act} to promote binary solutions. We iteratively solve \eqref{opt:activate} with Algorithm \ref{al:parallel1}. In the $m+1^\text{th}$ iteration of the Frank-Wolfe algorithm, the linearized quadratic upper-bound \eqref{eq:ub4} becomes
	\begin{align}
		\varrho_{s,1}^2(x_{s,1}^{m})x_{s,1}+\varrho_{s,2}^2(x_{s,2}^{m})x_{s,2},\label{eq:lin2}
	\end{align}
	where 
	\begin{align}
		&\varrho_{s,1}^2(x_{s,1}^{m})=\left(\frac{q}{(\underline{x}_{s,1}^{i}+\epsilon)^{1-q}}+e(x_{s,1}^{m}-\underline{x}_{s,1}^i)\right),\nonumber\\
		&\:\:\varrho_{s,2}^2(x_{s,2}^{m})=\left(\frac{q}{(\underline{x}_{s,2}^{i}+\epsilon)^{1-q}}+e(x_{s,2}^{m}-\underline{x}_{s,2}^i)\right). \nonumber
	\end{align}
	We exclude the linearized RHS of \eqref{eq:ub4}, i.e., \eqref{eq:lin2}, in the first $I$ steps of applying Algorithm \ref{al:parallel1} to solve \eqref{opt:activate}. Therefore, we set $\gamma^i=0$ when $i\leq I$. Next, we increase $\gamma^i$ iteratively and continue with $\gamma^{i+1}\geq \gamma^i>0$ when $i>I$. As the upper-bound in the RHS of \eqref{eq:ub4} iteratively receives a higher weight, the solutions become closer to $0$ or $1$. The summary of the proposed approach to solve \eqref{opt:activate} is given in Algorithm \ref{al:parallel2}. We iteratively continue solving with Algorithm \ref{al:parallel2}  until each $x_{s,1}^m$ and $x_{s,2}^m$ becomes binary.
	\begin{algorithm}[t!]
		0.	\textbf{Initialization} Assign values to $\underline{\mathbf{x}}^{0}$, $i=0$; \\
		\Repeat{$\underline{\mathbf{x}}^{i}$ becomes binary}{
			1. Apply Algorithm \ref{al:parallel1} to find $\underline{\mathbf{x}}^{i+1}$\;
			2. Update $\gamma^{i+1}$\;
			3. $i=i+1$\;
		}     
		\caption{Proposed algorithm to solve \eqref{opt:activate}}
		\label{al:parallel2}
	\end{algorithm}
	\section{The Distributed and Scalable Algorithm for Slice Reconfiguration}\label{sec:slrec}
	Problem \eqref{opt:slice1} is difficult to solve for two reasons: 1) $\ell_0$-norm is neither convex nor continuous; and 2)  \eqref{eq:out2} is not necessarily a convex constraint in $r_k^p$ and $t_k^p$ for an arbitrary $z_{k,S_n}^p(v_k^p,t_k^p)$. To tackle the difficulty of solving \eqref{opt:slice1}, we copy variables as $\mathbf{r}=\mathbf{g}$ and decouple  $\norm{\{\norm{\mathbf{r}_s-\underline{\mathbf{r}}_s^{n-1}}_2\}_{s \in\mathcal{S}_j,j\in\mathcal{J}}}_0$ and  \eqref{eq:out2} from \eqref{opt:slice1}. To alleviate the non-convexity of \eqref{opt:slice1} due to $\ell_0$-norm, we use the group LASSO regularization \cite{candes2008enhancing} with copied  $\mathbf{g}$. 
	We substitute 
	\begin{eqnarray}
		\sum_{j=1}^{J}\sum_{s\in \mathcal{S}_j}c_r(a_s^{1,i}\norm{\mathbf{g}_s-\underline{\mathbf{r}}_s^{n-1}}_2+a_s^{2,i}\norm{\mathbf{t}_s-\underline{\mathbf{t}}_s^{n-1}}_2),\label{eq:lasso}
	\end{eqnarray}
	for $\ell_0$-norms and iteratively minimize \eqref{eq:lasso}, where $i$ is the iteration counter. The update rules for coefficients in the $i^{\text{th}}$ iteration are 
	\begin{align}
		a_s^{1,i}=\frac{a_s^{1,0}}{\norm{\mathbf{g}_s^{i-1}-\underline{\mathbf{r}}_s^{n-1}}_2+\epsilon_1}\:\: \text{and}\:\:a_s^{2,i}=\frac{a_s^{2,0}}{\norm{\mathbf{t}_s^{i-1}-\underline{\mathbf{t}}_s^{n-1}}_2+\epsilon_2},\label{eq:updatelasso}
	\end{align}
	where $a_s^{1,0},a_s^{2,0}>0$. Furthermore, $\epsilon_1$ and $\epsilon_2$ are small positive numbers to ensure that a zero-valued norm in the denominator does not strictly prohibit $\norm{\mathbf{g}_s^i-\underline{\mathbf{r}}_s^{n-1}}_2=0$ in the next step.

	Using the group LASSO regularization, \eqref{opt:slice1} can be approximated by the following problem: 
	\begin{align}
		& \underset{\mathbf{r},\mathbf{t},\mathbf{g}\geq \mathbf{0}}{\text{min}}
		& & \hspace{-.2cm}\Gamma(\mathbf{r})+\sum_{j=1}^{J}\sum_{s\in \mathcal{S}_j}\hspace{-.1cm}c_ra_s^{1,i}\norm{\mathbf{g}_s-\underline{\mathbf{r}}_s^{n-1}}_2\label{opt:slice2}\nonumber\\
		&&&+\sum_{j=1}^{J}\sum_{s\in \mathcal{S}_j}\hspace{-.1cm}c_ra_s^{2,i}\norm{\mathbf{t}_s-\underline{\mathbf{t}}_s^{n-1}}_2\\
		& \text{s.t.}
		& & \eqref{eq:linkcap}, \eqref{eq:nodecap}, \eqref{eq:slicereq1}, \eqref{eq:slicereq2}, \eqref{eq:linkres}, \eqref{eq:noderes}, \eqref{eq:out2}, \mathbf{r}=\mathbf{g}.\nonumber
	\end{align}
	To solve each subproblem of the group LASSO \eqref{opt:slice2}, we propose an ADMM algorithm. 
	We dualize the constraint $
	\mathbf{r}=\mathbf{g}$ and find the augmented Lagrangian as follows:
	\allowdisplaybreaks
	\begin{align}\textstyle
		&L(\mathbf{r},\mathbf{t},\mathbf{g},\boldsymbol{\tau})=\sum_{j=1}^{J}\sum_{k \in\mathcal{K}_j}\sum_{p\in \mathcal{P}_k}\hspace{-.15cm}\left(\tau_k^p(r_k^{p} - g_k^{p})+\frac{\rho}{2}(r_k^{p} - g_k^{p})^2\right)\nonumber\\
		&+\Gamma(\mathbf{r})+\sum_{j=1}^{J}\sum_{s\in \mathcal{S}_j}\hspace{-.1cm}c_r(a_s^{1,i}\norm{\mathbf{g}_s-\underline{\mathbf{r}}_s^{n-1}}_2\hspace{-.1cm}+a_s^{2,i}\norm{\mathbf{t}_s-\underline{\mathbf{t}}_s^{n-1}}_2),\label{eq:lag}
	\end{align}
	where $\tau_k^p$ is a Lagrange multiplier and $\rho$ is a penalty parameter. We substitute \eqref{eq:lag} in the objective function of \eqref{opt:slice2} and alternatively minimize with respect to $\mathbf{r}$ as one block and $\{\mathbf{g},\mathbf{t}\}$ as the second block. We optimize with respect to both blocks in the $\pi+1^{\text{th}}$ iteration and  do $\boldsymbol{\tau}^{\pi+1}=\boldsymbol{\tau}^{\pi}+\rho(\mathbf{r}^{\pi+1}-\mathbf{g}^{\pi+1})$.

	\subsection{Subproblem with respect to $\mathbf{r}$}\label{sec:r}
	The subproblem with respect to $\mathbf{r}$ is as follows:
	\begin{align}
		& \underset{\mathbf{r}\geq \mathbf{0}}{\text{min}}
		& & \hspace{-.0cm}\Gamma(\mathbf{r})+\sum_{j=1}^{J}\sum_{k \in\mathcal{K}_j}\sum_{p\in \mathcal{P}_k}\left(\tau_k^{p,\pi}(r_k^{p} - g_k^{p,\pi})+\frac{\rho}{2}(r_k^{p} - g_k^{p,\pi})^2\right)\nonumber\\
		& \text{s.t.}
		& & \eqref{eq:linkcap}\label{opt:subr}.
	\end{align}\begin{figure*}[t]
		\begin{align}\textstyle
			&\underbrace{\sum_{j=1}^{J}\sum_{s \in\mathcal{S}_j}\int_{\sum_{k \in\mathcal{K}_s}\sum_{p\in \mathcal{P}_{k}}\hat{r}_{k}^{p}}^{\infty}\hspace{-.4cm}(y-\sum_{k \in\mathcal{K}_s}\sum_{p\in \mathcal{P}_{k}}\hat{r}_{k}^{p})f_{s,S_n}(y)dy}_{\Upsilon_1} - \sum_{j=1}^{J}\sum_{s \in\mathcal{S}_j} (1-F_{s,S_n}(\sum_{k \in\mathcal{K}_s}\sum_{p\in \mathcal{P}_{k}}\hat{r}_{k,s}^p))(\sum_{k \in\mathcal{K}_s}\sum_{p\in \mathcal{P}_{k}}r_{k}^{p}-\sum_{k \in\mathcal{K}_s}\sum_{p\in \mathcal{P}_{k}}\hat{r}_{k}^{p}) \nonumber \\
			&+\sum_{j=1}^{J}\sum_{s \in\mathcal{S}_j}\frac{\sum_{k \in\mathcal{K}_s}|\mathcal{P}_k|}{2}\sum_{k \in\mathcal{K}_s}\sum_{p\in \mathcal{P}_{k}}(r_{k}^{p}-\hat{r}_{k}^{p})^2\underbrace{-\sum_{j=1}^{J}\psi_j\Bigg(\sum_{s \in\mathcal{S}_j}\sum_{k \in\mathcal{K}_s}\mathbb{E}_{d_k}\Bigg[\phi_{k,S_n}\Big(\min(\sum_{p \in \mathcal{P}_k} \hat{r}_k^{p},d_k)\Big)\Bigg]\Bigg)}_{\Upsilon_2}\nonumber\\
			&+\sum_{j=1}^{J}\psi_j\left(\sum_{s \in\mathcal{S}_j}\sum_{k \in\mathcal{K}_s}\left(-\phi'_{k,S_n}\Big(\sum_{p \in \mathcal{P}_k} \hat{r}_k^{p}\Big)\Big(1-F_{k,S_n}(\sum_{p \in \mathcal{P}_k} \hat{r}_k^{p})\right)\left(\sum_{p\in \mathcal{P}_{k}}(r_k^{p}-\hat{r}_k^{p})\right)+\frac{Q}{2}\sum_{s \in\mathcal{S}_j}\sum_{k \in\mathcal{K}_s}\sum_{p\in \mathcal{P}_{k}}(r_k^{p}-\hat{r}_k^{p})^2\right).\label{eq:upper2}
		\end{align}
		\hrulefill
	\end{figure*}\normalsize
	To solve \eqref{opt:subr}, we substitute the quadratic upper-bound \cite[eq. (12)]{hong2015unified} for the expected excessive demand from slice $s$ given in \eqref{eq:iso}, which is convex with a Lipschitz continuous gradient. The Lipschitz constant for the gradient of \eqref{eq:iso} is $\sum_{k \in\mathcal{K}_s}|\mathcal{P}_k|$.  Similarly, we substitute a quadratic upper-bound for \eqref{eq:util1}. If $Q$ is the Lipschitz constant for the gradient of \eqref{eq:util1} ($Q$ can be calculated through differentiation of \eqref{eq:util1} and using extremum values of the first and second derivatives of $\phi_{k,S_n}(\cdot)$), then \eqref{eq:upper2} is an upper-bound for $\Gamma(\mathbf{r})$ around the current iterate $\hat{\mathbf{r}}$. In \eqref{eq:upper2}, $\Upsilon_1$ and $\Upsilon_2$ are constants and we consider $\Upsilon=\Upsilon_1+\Upsilon_2$. We notice that \eqref{eq:upper2} can be written in the following shorter form:
	\begin{align}
		&\sum_{j=1}^{J}\sum_{s \in\mathcal{S}_j}(\frac{\psi_jQ}{2}+\frac{\sum_{k \in\mathcal{K}_s}|\mathcal{P}_k|}{2})\sum_{k \in\mathcal{K}_s}\sum_{p\in \mathcal{P}_{k}}(r_k^{p}-\hat{r}_k^{p})^2\nonumber\\
		&+\sum_{j=1}^{J}\sum_{s \in\mathcal{S}_j}\sum_{k \in\mathcal{K}_s}\sum_{p\in \mathcal{P}_{k}}h_{k}^p(r_k^{p}-\hat{r}_k^{p})+\Upsilon,\label{eq:upsim}
	\end{align}
	where $h_{k}^p$ is a constant coefficient. When \eqref{eq:upsim} is substituted for $\Gamma(\mathbf{r})$ in \eqref{opt:subr}, \eqref{opt:subr} becomes a QP with a separable objective function in $\{r_k^p\}_{p\in\mathcal{P}_k,k=1:K}$. With the quadratic upper-bound \eqref{eq:upsim}, we can solve \eqref{opt:subr} using CPLEX or Gurobi or in a distributed fashion using the approach given in \cite{Reyhanian1,reyhanian2021resource}, which decomposes the problem across backhaul links. We associate one Lagrange multiplier $\mu_l$ to each constraint in \eqref{eq:linkcap} and find the Lagrangian as follows:
	\begin{align}
		&L_c(\mathbf{r},\boldsymbol{\mu})=\sum_{j=1}^{J}\sum_{s \in\mathcal{S}_j}(\frac{\psi_jQ}{2}+\frac{\sum_{k \in\mathcal{K}_s}|\mathcal{P}_k|}{2})\sum_{k \in\mathcal{K}_s}\sum_{p\in \mathcal{P}_{k}}(r_k^{p}-\hat{r}_k^{p})^2\nonumber\\
		&+\sum_{j=1}^{J}\sum_{s \in\mathcal{S}_j}\sum_{k \in\mathcal{K}_s}\sum_{p\in \mathcal{P}_{k}}h_{k}^p(r_k^{p}-\hat{r}_k^{p})+\Upsilon\nonumber\\
		&+\sum_{l\in \mathcal{L}}\mu_l\left(\sum_{j=1}^{J}\sum_{k\in\mathcal{K}_j}\sum_{p:\{p\in \mathcal{P}_k,l\in\mathcal{L}_k^p\}}r_k^{p} - C_l\right).\nonumber
	\end{align}
	We can decompose the Lagrangian across backhaul links and parallelelize the computation across links as $L_c(\mathbf{r},\boldsymbol{\mu})$ is strongly convex in $\mathbf{r}$ and has a Lipschitz continuous gradient \cite[Theorem 1]{reyhanian2021resource}. Each term in the Lagrangian that includes $r_k^{p}$ is decomposed across links that construct path $p$. We denote the Lagrange multiplier in the $m^{\text{th}}$ iteration of the distributed approach by $\widetilde{\mu}_l^{m}$. In the $m^{\text{th}}$ iteration, each link of path $p$ receives a portion of 
	\begin{align}
		\alpha_{k,l}^{p,m}=\frac{\widetilde{\mu}_l^{m-1}}{\sum_{l'\in \mathcal{L}_k^p}\widetilde{\mu}_{l'}^{m-1}},\label{eq:alpha}
	\end{align}
	from those Lagrangian terms that include $r_k^{p}$. The decomposed Lagrangian on a backhaul link is as follows:\allowdisplaybreaks
	\begin{align}
		&L_l(\mathbf{r}_l,\mu_l,\widetilde{\boldsymbol{\mu}}^{m-1})=\sum_{j=1}^{J}\sum_{k \in\mathcal{K}_j}\sum_{p:\{p\in \mathcal{P}_k,l\in\mathcal{L}_k^p\}}\alpha_{k,l}^{p,m}\tau_k^{p,\pi}(r_k^{p} - g_k^{p,\pi})\nonumber\\
		&+\sum_{j=1}^{J}\sum_{k \in\mathcal{K}_j}\sum_{p:\{p\in \mathcal{P}_k,l\in\mathcal{L}_k^p\}}\frac{\rho}{2}\:\alpha_{k,l}^{p,m}(r_k^{p} - g_k^{p,\pi})^2\nonumber\\
		&+\sum_{j=1}^{J}\sum_{s \in\mathcal{S}_j}(\frac{\psi_jQ}{2}+\frac{\sum_{k \in\mathcal{K}_s}|\mathcal{P}_k|}{2})\sum_{k \in\mathcal{K}_s}\sum_{p:\{p\in \mathcal{P}_k,l\in\mathcal{L}_k^p\}}\hspace{-0.7cm}\alpha_{k,l}^{p,m}(r_k^{p}-\hat{r}_k^{p})^2\nonumber\\
		&+\sum_{j=1}^{J}\sum_{k \in\mathcal{K}_j}\sum_{p:\{p\in \mathcal{P}_k,l\in\mathcal{L}_k^p\}}\alpha_{k,l}^{p,m}h_k^p(r_k^{p}-\hat{r}_k^{p})\nonumber\\
		&+\mu_l\left(\sum_{j=1}^{J}\sum_{k\in\mathcal{K}_j}\sum_{p:\{p\in \mathcal{P}_k,l\in\mathcal{L}_k^p\}}r_k^{p} - C_l\right)\nonumber,
	\end{align}
	where $\alpha_{k,l}^{p,m}$ is a constant and $\mathbf{r}_l=\{r_k^{p}\}_{p\in \mathcal{P}_k,l\in\mathcal{L}_k^p,k=1:K}$. 
	The optimal $r_k^{p}$ and $\mu_l$ are obtained from the KKT conditions for the per-link subproblem as follows:
	\begin{subequations}
		\begin{align}
			&\frac{\partial L_l(\mathbf{r}_l,\mu_l,\widetilde{\boldsymbol{\mu}}^{m-1})}{\partial r_k^{p}}=0,r_k^p\geq 0, \mu_l \geq 0,\hspace{0.6cm}p\in \mathcal{P}_k,l\in\mathcal{L}_k^p,\forall k,\label{eq:kkt1}\\
			&
			\mu_l\left(\sum_{k=1}^K\sum_{p:\{p\in \mathcal{P}_k,l\in\mathcal{L}_k^p\}}r_k^{p} - C_l\right)=0,\hspace{1.2cm} \label{eq:kkt2}\\
			&
			\sum_{k=1}^K\sum_{p:\{p\in \mathcal{P}_k,l\in\mathcal{L}_k^p\}}r_k^{p} \leq C_l,\hspace{1.2cm} \label{eq:kkt3}
		\end{align}
	\end{subequations}
	The solution to \eqref{eq:kkt1} can be obtained in closed-form for the $l^{\text{th}}$ backhaul link as follows:
	\begin{align}
		&r_k^{p}\hspace{-0.1cm}=\max\Big(0,\frac{\alpha_{k,l}^{p,m}\big((\psi_jQ+\sum_{k \in\mathcal{K}_s}\hspace{-0.2cm}|\mathcal{P}_k|)\hat{r}_k^{p}+\rho  g_k^{p,\pi}\hspace{-0.2cm}-h_k^p-\tau_k^{p,\pi}\big)\hspace{-0.1cm}-\mu_l}{\alpha_{k,l}^{p,m}(\rho+\psi_jQ+\sum_{k \in\mathcal{K}_s}|\mathcal{P}_k|)}\Big).\label{eq:rsol}
	\end{align}
	To find the optimal $r_k^{p}$ and $\mu_l$ for each per-link subproblem, we implement a bisection search on $\mu_l$ in the positive orthant and obtain the corresponding $r_k^{p}$ from \eqref{eq:rsol} until one $\mu_l$ that satisfies the complementary slackness condition \eqref{eq:kkt2} is obtained. If there is no such $\mu_l$, then we set $\widetilde{\mu}_l^{m}=0$ and $\alpha_{k,l}^{p,m+1}=0$. For these links, we do not need to continue computation as the KKT conditions remain satisfied. In the following iterations, we ignore these links and consider links with $\widetilde{\mu}_l^{m} > 0$. We find $\mu_l$ in parallel for all links. Once each $\mu_l$ is obtained, we set $\widetilde{\mu}_l^m=\mu_l,\widetilde{r}_k^{p,m}=r_k^{p}$ and update $\alpha_{k,l}^{p,m+1}=\widetilde{\mu}_l^{m}/\sum_{l'\in \mathcal{L}_k^p}\widetilde{\mu}_{l'}^{m}$. We alternate between finding $\widetilde{\mu}_l^{m}$ and $\alpha_{k,l}^{p,m+1}$ until each $\widetilde{\mu}_l^{m}$ converges. Upon convergence, we update $\hat{\mathbf{r}}=\widetilde{\mathbf{r}}^m$. We repeat this until $\hat{\mathbf{r}}$ converges. The summary of the proposed approach to solve \eqref{opt:subr} is given in Algorithm \ref{al:parallel}.
	
	\begin{algorithm}[t!]
		0.	\textbf{Initialization} Assign some small positive number to each $\widetilde{\mu}_l^0$; \\
		\For{a few iterations}{
			1. Assign a positive number to each $\widetilde{\mu}_l^0$, $m=0$\;		
			\Repeat{all $\widetilde{\mu}_l^m$ variables converge}{
				\For{all links}{ 
					2. Find  $\alpha_{k,l}^{p,m+1}=\widetilde{\mu}_l^m/\sum_{l'\in \mathcal{L}_k^p}\widetilde{\mu}_{l'}^m$\;
					\If{$\widetilde{\mu}_l^{m} >0$}{3. Find $\widetilde{\mu}_l^{m+1}$ from \eqref{eq:kkt1} and \eqref{eq:kkt2}\;
					}
					4. $m=m+1$\;
				}
			}     
			\For{all $r_k^{p}$ variables}{5. Find $\widetilde{r}_k^{p,m}$ from \eqref{eq:rsol} for a per-link subproblem, where $l\in \mathcal{L}_k^p$ and $\widetilde{\mu}_l^m>0$\;}
			6. $\hat{\mathbf{r}}=\widetilde{\mathbf{r}}^m$\;
		}
		\caption{The proposed approach to solve \eqref{opt:subr}}
		\label{al:parallel}
	\end{algorithm}
	
	We do successive upper-bound minimization \cite{razaviyayn2013unified} and solve \eqref{opt:subr} for a few iterations with Algorithm \ref{al:parallel} and update $\hat{\mathbf{r}}$ in each iteration. Then, we update $\mathbf{r}^{\pi+1}=\hat{\mathbf{r}}$. As the quadratic upper-bound \eqref{eq:upper2} satisfies the four convergence conditions for block successive upper-bound minimization methods \cite[Assumption 2]{razaviyayn2013unified}, the successive upper-bound minimization converges to the KKT solution of \eqref{opt:subr}.

	\subsection{Subproblem with respect to $\mathbf{g}$ and $\mathbf{t}$}\label{sec:gt}
	
	The $\ell_2$-norms $\norm{\mathbf{g}_s-\underline{\mathbf{r}}_s^{n-1}}_2$ and $\norm{\mathbf{t}_s-\underline{\mathbf{t}}_s^{n-1}}_2$ are not smooth. Instead of these two, we consider the following smooth approximation with a Lipschitz continuous gradient:
	\begin{align}
		&U(\mathbf{g}_s,\mathbf{t}_s)=-2\delta c_r(a_s^{1,i}+a_s^{2,i})\nonumber\\
		&+c_ra_s^{1,i}\sqrt{\sum_{l \in \mathcal{L}}(\sum_{k\in\mathcal{K}_s}\sum_{p:\{p\in \mathcal{P}_k,l\in\mathcal{L}_k^p\}}\hspace{-.6cm}g_k^{p}-\underline{r}_s^{l,n-1 })^2+\delta^2}\nonumber\\
		&+c_ra_s^{2,i}\sqrt{\sum_{b \in \mathcal{B}}(\sum_{k \in\mathcal{K}_s}\sum_{p:\{p\in \mathcal{P}_k,b \in \mathcal{U}_k^p\}}\hspace{-.6cm}t_k^{p}-\underline{t}_s^{b,n-1})^2+\delta^2},\nonumber
	\end{align} \normalsize
	where $\delta$ is a small positive number. We note that the minima of $c_r(a_s^{1,i}\norm{\mathbf{g}_s-\underline{\mathbf{r}}_s^{n-1}}_2+a_s^{2,i}\norm{\mathbf{t}_s-\underline{\mathbf{t}}_s^{n-1}}_2)$ is identical to that of $U(\mathbf{g}_s,\mathbf{t}_s)$, which is $\{\mathbf{g}_s,\mathbf{t}_s\}=\{\underline{\mathbf{r}}_s^{n-1},\underline{\mathbf{t}}_s^{n-1}\}$.
	The decomposed, per slice subproblem with respect to $\mathbf{g}_s$ and $\mathbf{t}_s$ becomes:
	\allowdisplaybreaks
	\begin{subequations}\label{opt:g}
		\begin{align}
			& \hspace{-.1cm}\underset{\mathbf{g}_s,\mathbf{t}_s\geq \mathbf{0}}{\text{min}}
			& & \hspace{-.2cm}U(\mathbf{g}_s,\mathbf{t}_s)+(\mathbf{r}_s^{\pi+1} - \mathbf{g}_s)'\boldsymbol{\tau}_s^{\pi}+\frac{\rho}{2}\norm{\mathbf{r}_s^{\pi+1} - \mathbf{g}_s}_2^2\label{opt:g1}\\
			& \hspace{.2cm}\text{s.t.}& &\hspace{-.5cm}\int_{0}^{g_k^p}\hspace{-.2cm}z_{k,S_n}^p(v_k^p,t_k^{p})\:(g_k^{p}-v_k^p)dv_k^p\leq \beta_sg_k^{p}, p\in \mathcal{P}_k, k\in \mathcal{K}_s,\label{opt:g2}\\
			&&&\hspace{-.4cm}\eqref{eq:nodecap},\eqref{eq:slicereq1},\eqref{eq:slicereq2}.
		\end{align}
	\end{subequations}
	In the above problem, $\mathbf{g}_s=\{g_s^l\}_{l\in\mathcal{L}}$ and $\mathbf{t}_s=\{t_s^b\}_{b \in \mathcal{B}}$. We can rewrite \eqref{opt:g1} as a function of $\{g_k^p,t_k^p\}$ if we substitute $\{g_k^p\}_{p\in \mathcal{P}_k,k\in\mathcal{K}_s}$ for $\mathbf{g}_s$ and substitute $\{t_k^p\}_{p\in \mathcal{P}_k,k\in\mathcal{K}_s}$ for $\mathbf{t}_s$. The above problem is not separable in $\mathbf{g}_s$ and $\mathbf{t}_s$ due to \eqref{opt:g2}. To alleviate the complexity of solving the above problem, we decouple \eqref{opt:g2} from \eqref{opt:g} by copying optimization variables and using ADMM. We add the constraints $\mathbf{o}=\mathbf{g}$ and $\mathbf{f}=\mathbf{t}$. We dualize these two constraints and find the augmented Lagrangian. In the $m+1^{\text{th}}$ iteration of the inner ADMM, the subproblem with respect to $\mathbf{g}_s$ and $\mathbf{t}_s$ becomes:
	\begin{align}
		& \underset{\mathbf{g}_s,\mathbf{t}_s}{\text{min}}
		& & \hspace{-.1cm}(\mathbf{r}_s^{\pi+1} - \mathbf{g}_s)'\boldsymbol{\tau}_s^\pi+\frac{\rho}{2}\norm{\mathbf{r}_s^{\pi+1}- \mathbf{g}_s}_2^2+(\mathbf{o}_s^m - \mathbf{g}_s)'\boldsymbol{\iota}_s^m\nonumber\\
		&&&\hspace{-0.2cm}+\frac{\rho_{11}}{2}\norm{\mathbf{o}_s^m - \mathbf{g}_s}_2^2+(\mathbf{f}_s^m - \mathbf{t}_s)'\boldsymbol{\lambda}_s^m+\frac{\rho_{12}}{2}\norm{\mathbf{f}_s^m - \mathbf{t}_s}_2^2\nonumber\\
		&&&+U(\mathbf{g}_s,\mathbf{t}_s)\nonumber\\
		& \text{s.t.}& &\hspace{-.cm}\eqref{eq:nodecap},\eqref{eq:slicereq1},\eqref{eq:slicereq2},\mathbf{g}_s,\mathbf{t}_s\geq \mathbf{0}.\label{opt:gg}
	\end{align}\normalsize
	\begin{figure}[t!]
		\centering
		\includegraphics[width=.35\textwidth]{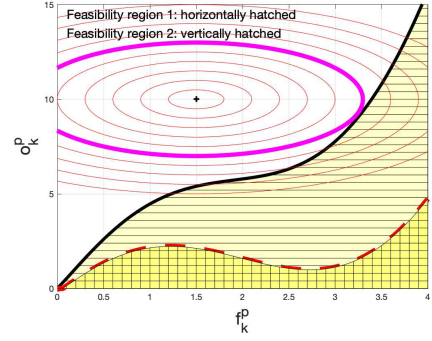}
		\caption{The feasibility region and level sets.}\label{fig:quasiconvex}
	\end{figure} 
	\hspace{-.15cm}Problem \eqref{opt:gg} is separable and strongly convex in $\mathbf{g}_s$ and $\mathbf{t}_s$.
	To separately optimize with respect to $\mathbf{g}_s$ and $\mathbf{t}_s$, we can use the proximal gradient descent algorithm or existing solvers. 
	Then, we update $\{\mathbf{g}_s^{m+1},\mathbf{t}_s^{m+1}\}$. Therefore, the subproblem with respect to $\mathbf{o}_s$ and $\mathbf{f}_s$ is
	\begin{align}
		& \underset{\mathbf{o}_s,\mathbf{f}_s\geq \mathbf{0}}{\text{min}}
		& & \hspace{-.2cm}
		(\mathbf{f}_s - \mathbf{t}_s^{m+1})'\boldsymbol{\lambda}_s^{m}+\frac{\rho_{12}}{2}\norm{\mathbf{f}_s - \mathbf{t}_s^{m+1}}_2^2\label{opt:tt}\\
		&&&\hspace{-1cm}+(\mathbf{o}_s - \mathbf{g}_s^{m+1})'\boldsymbol{\iota}_s^{m}+\frac{\rho_{11}}{2}\norm{\mathbf{o}_s - \mathbf{g}_s^{m+1}}_2^2 \hspace{.6cm} \text{s.t.}\hspace{0.6cm}\eqref{opt:g2}.\nonumber
	\end{align}\normalsize
	We observe that \eqref{opt:tt} is decomposable across paths, and we find a separate subproblem for each $(o_k^p,f_k^p)$ pair. For each $o_k^p$, we can find the lowest $f_k^p$ that ensures \eqref{opt:g2}. Then, a feasible set for optimization variables is characterized (e.g., feasibility region $1$ in Fig. \ref{fig:quasiconvex}). Since the objective function of \eqref{opt:tt} is strongly convex, one can find the global solutions with respect to $\mathbf{o}_s$ and $\mathbf{f}_s$ by setting the gradient of the objective function to zero and projecting the obtained solutions to the positive orthant. If the obtained solutions satisfy \eqref{opt:g2}, then problem \eqref{opt:tt} is solved. If the minimizers of the objective function of \eqref{opt:tt} do not satisfy \eqref{opt:g2}, then we solve the problem using a bisection search. The bisection search works on level sets of the objective function of \eqref{opt:tt} (e.g., Fig. \ref{fig:quasiconvex}). For each level set considered by the bisection, we determine whether or not the considered level set intersects with the characterized feasibility set. This can be done by a simple search as follows. We consider the rightmost point and the bottommost point of the level set. With a sufficiently small constant step-size, we move to the right on the level set from the bottommost point. In each point, we measure the vertical distance between the level set and the border of the feasibility region. If the distance becomes zero in one point, then the level set intersects with the feasibility region. If the vertical distance variation is not monotone, there is a chance of intersection. Using the bisection search on the level sets, we find the level set that is tangent to the feasibility region. The obtained point of the feasibility region, which is touched by the tangent level set, is the minimizer of \eqref{opt:tt}. The summary of the proposed approach to solve \eqref{opt:tt} is given in Algorithm \ref{al:g}. 
	\begin{algorithm}[t!]
		0.	\textbf{Initialization} $A_1=0$, $A_2=\text{a large number}$\\
		\While{$|A_1-A_2|$ is not small enough}
		{1. $A_3=(A_1+A_2)/2$;\\
			2. Find the level sets of the objective function of \eqref{opt:tt} that correspond to $A_1, A_2$ and $A_3$;\\
			\eIf{the level set corresponding to $A_3$ intersects with the feasibility region}{3. $A_2\leftarrow A_3$;}{4. $A_1\leftarrow A_3$;}
		}
		5. Return the feasible point touched by the tangent level set;
		\caption{The proposed approach to solve \eqref{opt:tt}}
		\label{al:g}
	\end{algorithm}
	\begin{algorithm}[t!]
		0.	\textbf{Initialization} $\boldsymbol{\iota}^{0}=\mathbf{0}$, $\boldsymbol{\lambda}^{0}=\mathbf{0}$, $m=0$; \\
		\Repeat{the primal and dual residuals are small enough}{
			1. Use proximal gradient descent to solve \eqref{opt:gg}\;
			2. Use Algorithm \ref{al:g} to solve \eqref{opt:tt}\;
			3. $\boldsymbol{\iota}^{m+1}=\boldsymbol{\iota}^{m}+\rho_{11}(\mathbf{o}^{m+1} - \mathbf{g}^{m+1})$ and $\boldsymbol{\lambda}^{m+1}=\boldsymbol{\lambda}^{m}+\rho_{12}(\mathbf{f}^{m+1} - \mathbf{t}^{m+1})$\;
			4. $m=m+1$\;
		}
		\caption{The ADMM approach to solve \eqref{opt:g}}
		\label{al:gt}
	\end{algorithm}
	
	Once we obtain the optimal $\{\mathbf{o}_s,\mathbf{f}_s\}$, we update $\{\mathbf{o}_s^{m+1},\mathbf{f}_s^{m+1}\}$. We alternatively minimize with respect to $\{\mathbf{g}_s,\mathbf{t}_s\}$ and $\{\mathbf{o}_s,\mathbf{f}_s\}$. Then, we update Lagrange multipliers as $\boldsymbol{\iota}^{m+1}=\boldsymbol{\iota}^{m}+\rho_{11}(\mathbf{o}^{m+1} - \mathbf{g}^{m+1})$ and $\boldsymbol{\lambda}^{m+1}=\boldsymbol{\lambda}^{m}+\rho_{12}(\mathbf{f}^{m+1} - \mathbf{t}^{m+1})$. The summary of the proposed approach to solve \eqref{opt:g} is given in Algorithm \ref{al:gt}.  After Algorithm \ref{al:gt} converges (the primal and dual residuals are small enough, see \cite[p. 15--18]{Boyd_2010}), we use obtained $\mathbf{g}$ to update $\mathbf{g}^{\pi+1}$.
	
	\begin{proposition}\label{prop:tir1}
		Algorithm \ref{al:g} converges to the global stationary solution of the problem in \eqref{opt:tt} when the expected outage is non-increasing in the transmission resource.
	\end{proposition}
	\begin{proof} 
		Consider that we denote the expected outage in (8) by $\varphi_1(r,t)$. Based on the assumption in Proposition \ref{prop:tir1}, $\varphi_1(r,t)$ is non-increasing in $t$. We claim that the feasibility region characterized by (8) cannot narrow when $t$ increases (e.g., feasibility region $2$ in Fig. \ref{fig:quasiconvex}). We prove this claim by contradiction. Consider $t_1<t_2$ and the feasibility region narrows as $t$ increases, where $(r,t_1)$ belongs to the feasibility region while $(r,t_2)$ does not belong to the feasibility region. Then, we have $\varphi_1(r,t_1)\leq\beta_s r$ and $\varphi_1(r,t_2)>\beta_s r$. Based on this, we should have $\varphi_1(r,t_1)<\varphi_1(r,t_2)$. This is a contradiction as we assumed $\varphi_1(r,t)$ is non-increasing in the transmission resource. Therefore, the feasibility region cannot narrow as $t$ increases similar to feasibility region $1$ in Fig. \ref{fig:quasiconvex}. Consider that the border line that specifies the feasibility region is characterized by the following equation $\varphi_2(r)=t$, where $\varphi_2(r)$ is non-decreasing. Then, the feasibility region of (8) can be represented as $\varphi_2(r)\leq t$. We consider the following optimization problem:
		\begin{align}
			& \underset{r,t\geq 0 }{\text{min}}
			& & \hspace{-.cm}\varphi_3(r)+\varphi_4(t)\nonumber\\
			& \text{s.t.}& &\hspace{-.cm}\varphi_2(r)\leq t,\nonumber
		\end{align}\normalsize
		where 
		$\varphi_3(r)$ an $\varphi_4(t)$ are strongly convex functions with global minimas $r^\star$ and $t^\star$, respectively. The above problem is not always convex as $\varphi_2(r)$ is not always convex.

		We can find the level sets of the objective function from $\varphi_3(r)+\varphi_4(t)=\kappa$, where we change $\kappa$. The solution of the above problem is a point in the feasible set that gives the least objective function value. If the minimizer of the objective function lies in the feasible set, then that is the problem solution. Otherwise, the problem solution lies on the level set that is tangent to the feasible set and is exactly the point on the feasible set border line that is touched by the level set. As the gradient of each point on a level set is different from the gradient of the other points, only one point of the feasible set is touched by the tangent level set. Algorithm 4 can find the tangent level set using a bisection search since there is only one tangent level set, and the rest of level sets are either intersecting with the feasibility region or do not intersect. In each iteration, we quickly evaluate whether or not three considered level sets by the bisection method are intersecting the feasible set or not. Since the objective function of (49) is strongly convex, the obtained solution is unique and
		Algorithm 4 solves the problem globally.
	\end{proof}
	\begin{remark}
		If \eqref{eq:rewritte} is considered instead of \eqref{eq:out1}, then we rewrite \eqref{opt:g2} as
		\begin{align}
			\int_{0}^{g_k^w}z_{k,S_n}^p(v_k^w,t_k^{w})\:(g_k^{w}-v_k^w)dv_k^w\leq \beta_sg_k^{w}, k\in \mathcal{K}_s,\label{eq:pan}
		\end{align}
		where $g_k^w=\sum_{p:\{p\in \mathcal{P}_k,w\in p\}}g_k^{p}$. To solve \eqref{opt:tt} with \eqref{eq:pan}, we dualize $g_k^w=\sum_{p:\{p\in \mathcal{P}_k,w\in p\}}g_k^{p}$ and find the augmented Lagrangian. We can deploy an ADMM algorithm to solve the problem. We use Algorithm \ref{al:g} to solve the subproblem with respect to $\{\{g_k^w\}_{w\in\mathcal{W}_k,k=1:K},\{t_k^w\}_{w\in\mathcal{W}_k,k=1:K}\}$. Moreover, the augmented Lagrangian minimizers with respect to $\{g_k^p\}_{p\in\mathcal{P}_k,k=1:K}$ can be obtained in closed-form.
		\begin{algorithm}[t!]
			0.	\textbf{Initialization} $\boldsymbol{\tau}^{0}=\mathbf{0}$, $\pi=0$, assign values to $a_s^{1,0}$ and $a_s^{2,0}$, $i=1$; \\
			\Repeat{both $a_s^{1,i}$ and $a_s^{2,i}$ converge}{
				1. Find $a_s^{1,i}$ and $a_s^{2,i}$ from \eqref{eq:updatelasso}\;
				\Repeat{the primal and dual residuals are small enough and $\hat{\mathbf{r}}$ converges}{
					2. Apply Algorithm \ref{al:parallel} to obtain $\mathbf{r}^{\pi+1}$\;
					3. Apply Algorithm \ref{al:gt} to obtain $\mathbf{g}^{\pi+1}$\;
					4. $\boldsymbol{\tau}^{\pi+1}=\boldsymbol{\tau}^\pi+\rho(\mathbf{r}^{\pi+1}-\mathbf{g}^{\pi+1})$\;
					5. $\pi=\pi+1$\;
				}
				6. $i=i+1$\;
			}
			7.  $\underline{\mathbf{r}}_s^{n}=\mathbf{r}_s^{\pi}$ and $\underline{\mathbf{t}}_s^{n}=\mathbf{t}_s^{\pi}$
			\caption{The ADMM algorithm to solve \eqref{opt:slice1}}
			\label{al:admm}
		\end{algorithm}
	\end{remark}
	\begin{proposition}\label{prop:tir2}
		Algorithm \ref{al:gt} converges to the global stationary solution to the problem in \eqref{opt:g}.
	\end{proposition}
	\begin{proof}
		The first subproblem of Algorithm 5, given in (48), is strongly convex and has a unique global optimal solution. Furthermore, we proved in Proposition \ref{prop:tir1} that the unique global solution is obtained by Algorithm 4 for (49). Therefore, the obtained solution by Algorithm 5 is global stationary \cite[p. 698]{bertsekas1999nonlinear}.
	\end{proof}

	\subsection{The proposed ADMM Algorithm to solve \eqref{opt:slice1}}
	In the proposed ADMM algorithm, we alternatively optimize with respect to $\mathbf{r}$ as one block and jointly  with respect to $\{\mathbf{g},\mathbf{t}\}$ as the other block. In the proposed ADMM algorithm, $\mathbf{r}$ is updated through a block successive upper-bound minimization in Section \ref{sec:r}, while $\{\mathbf{g},\mathbf{t}\}$ is updated through an inner ADMM in Section \ref{sec:gt}. In the $\pi+1^{\text{th}}$ iteration of the ADMM algorithm, we update the Lagrange multipliers  as $\boldsymbol{\tau}^{\pi+1}=\boldsymbol{\tau}^\pi+\rho(\mathbf{r}^{\pi+1}-\mathbf{g}^{\pi+1})$. We continue alternating between optimization with respect to $\mathbf{r}$ and $\{\mathbf{g},\mathbf{t}\}$ until both variable blocks and $\hat{\mathbf{r}}$ converge. We solve each subproblem of the group LASSO regularization using the proposed ADMM algorithm. We then update $a_s^{1,i+1}$ and $a_s^{2,i+1}$ using \eqref{eq:updatelasso} and solve again. We continue until both $a_s^{1,i}$ and $a_s^{2,i}$ converge.  The summary of the proposed approach to solve \eqref{opt:slice1} is given in Algorithm \ref{al:admm}. 
		\begin{figure*}[t!]
		\centering
		\begin{subfigure}{0.32\textwidth}
			\includegraphics[width=\textwidth]{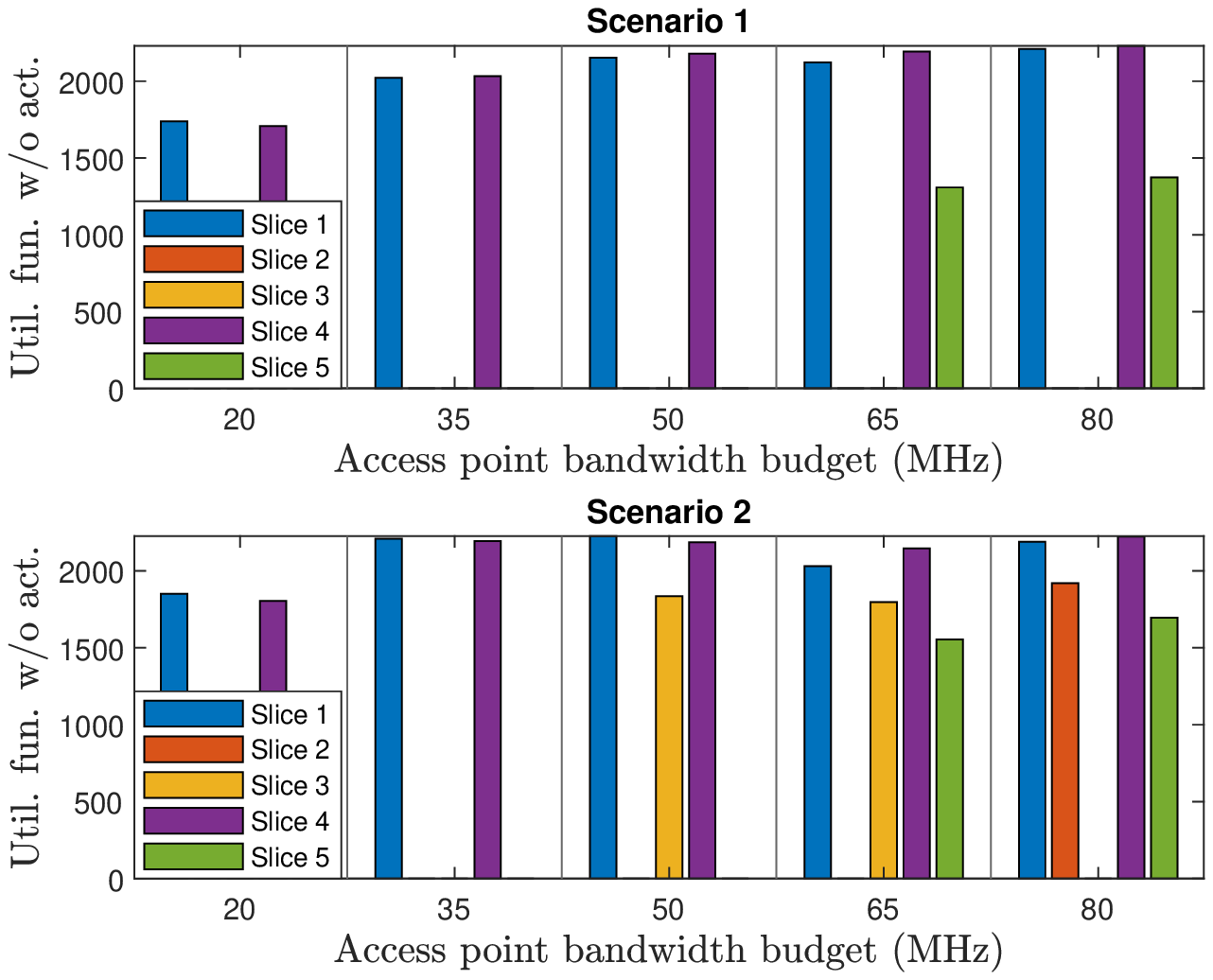}
			\caption{}
			\label{fig:act_obj}
		\end{subfigure}
		\begin{subfigure}{0.32\textwidth}
			\includegraphics[width=\textwidth]{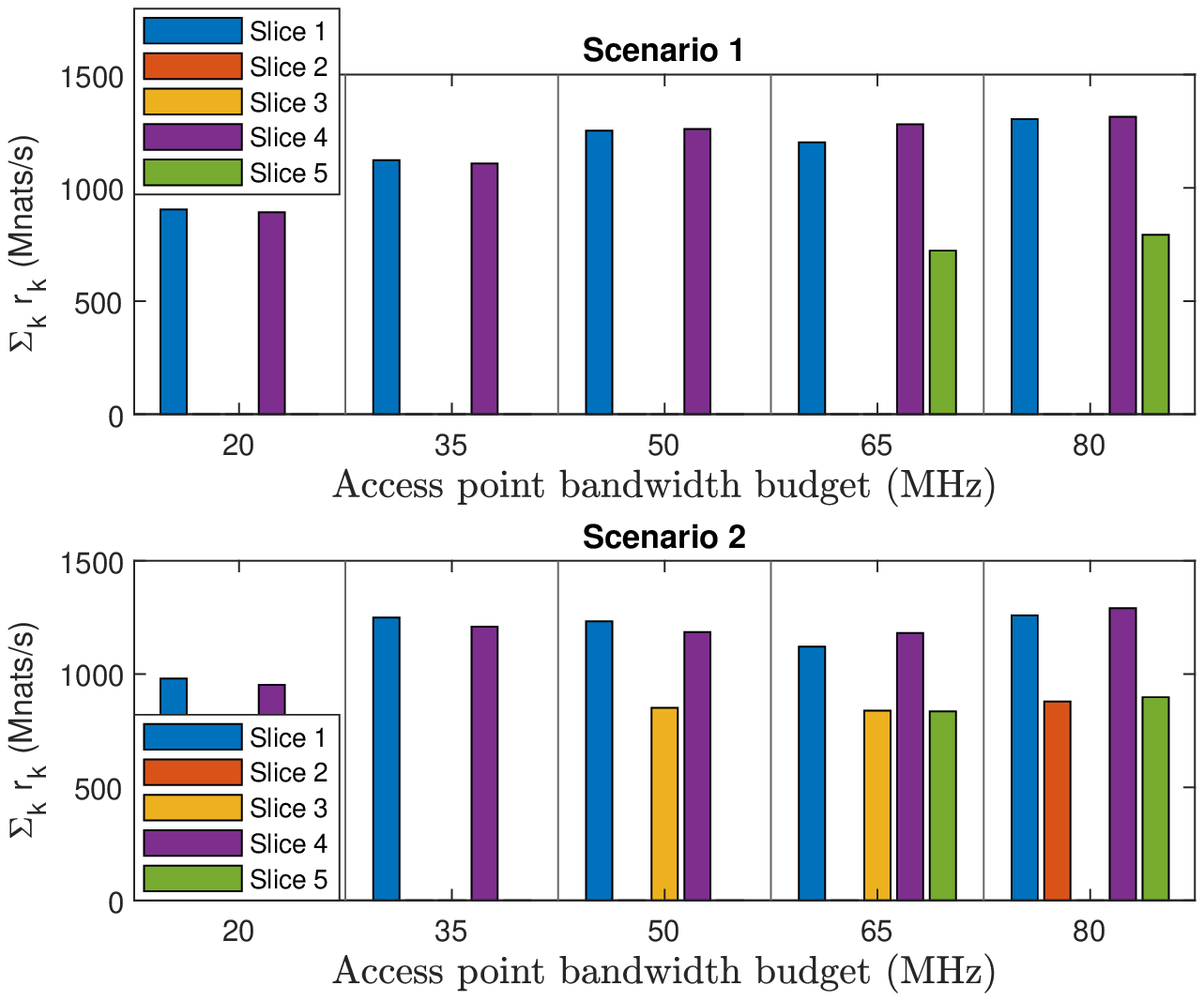}
			\caption{}
			\label{fig:act_rate}
		\end{subfigure}\\
		\begin{subfigure}{0.32\textwidth}
			\includegraphics[width=\textwidth]{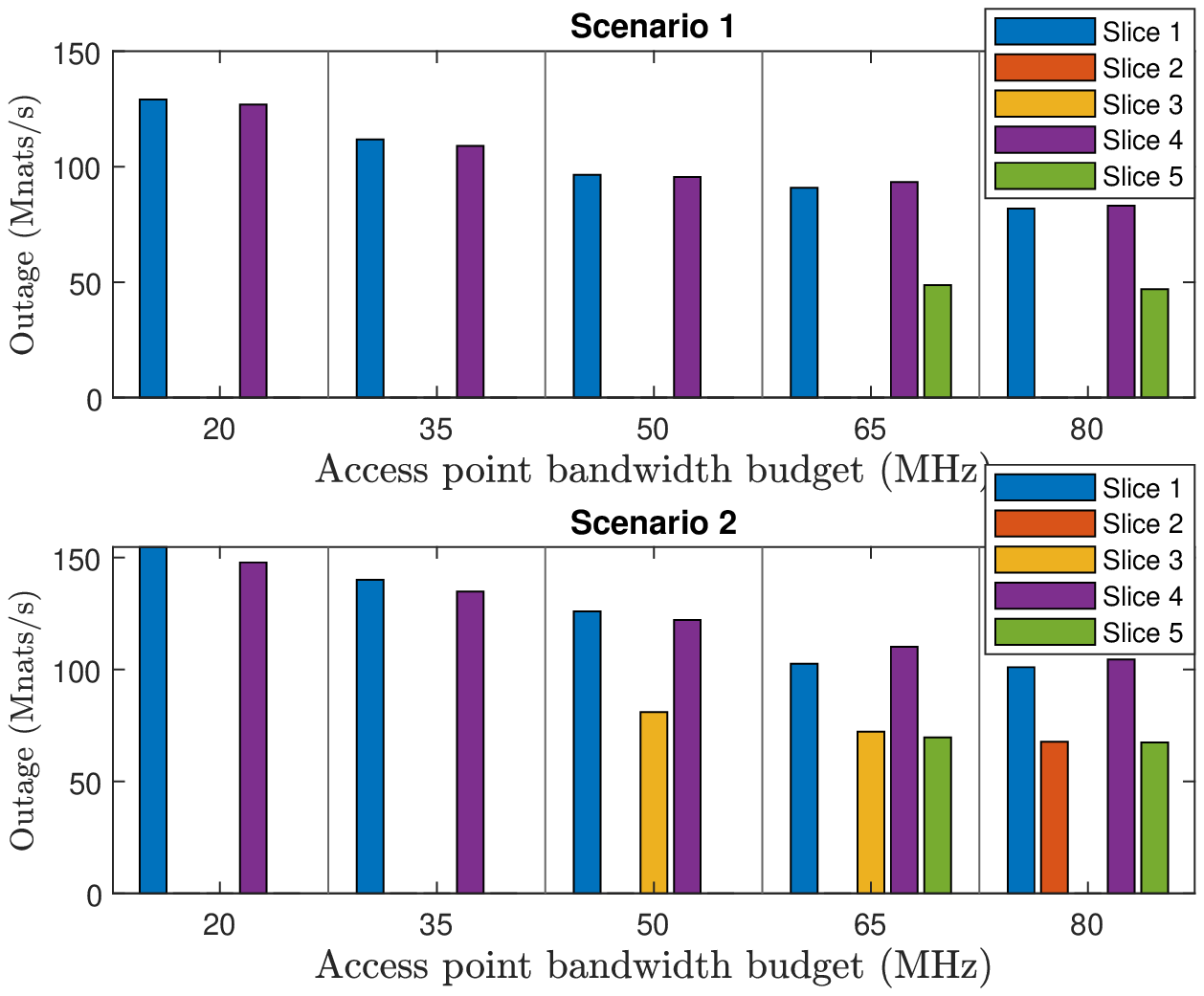}
			\caption{}
			\label{fig:act_outage}
		\end{subfigure}
		\begin{subfigure}{0.32\textwidth}
			\includegraphics[width=\textwidth]{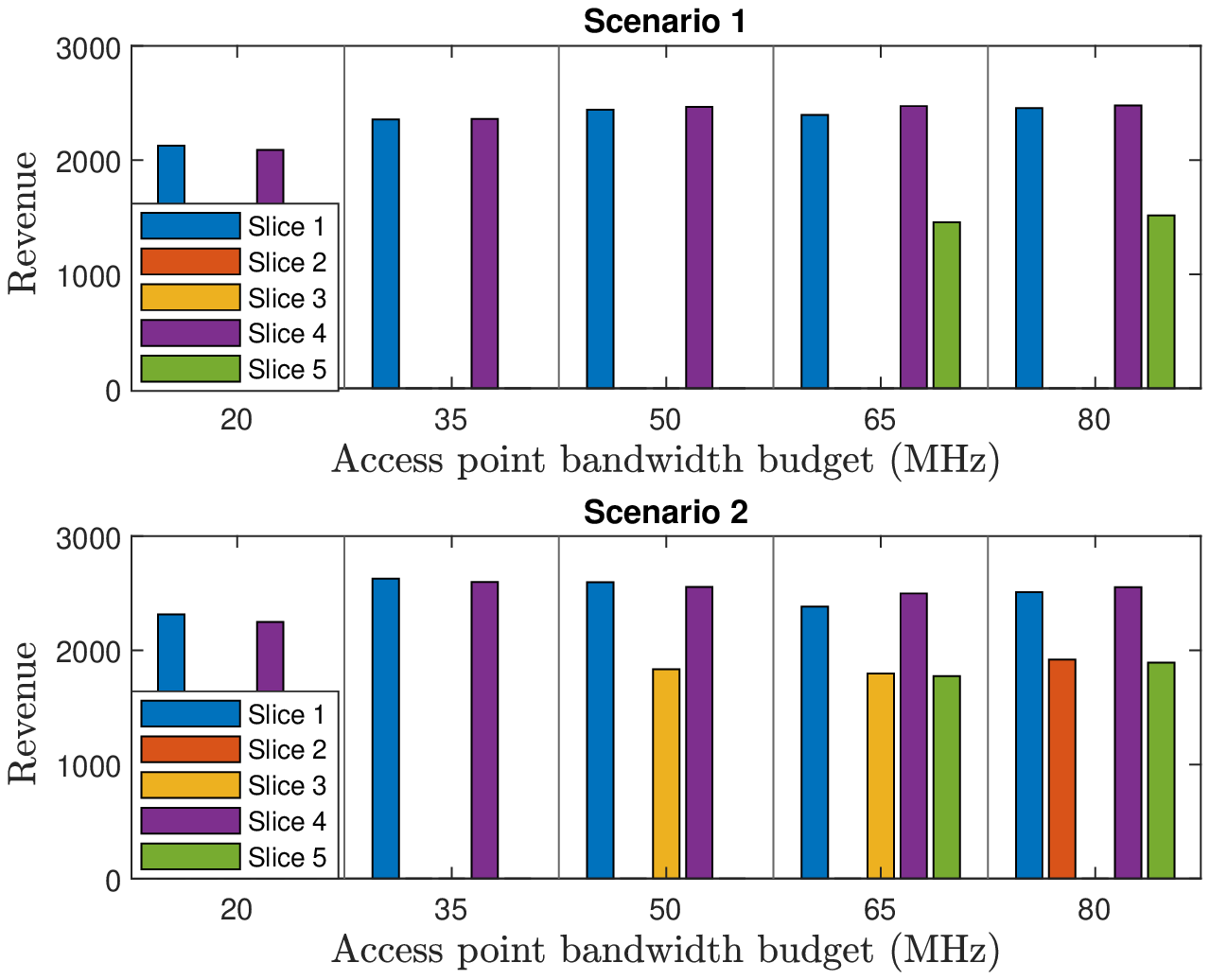}
			\caption{}
			\label{fig:act_rev}
		\end{subfigure}
		\caption{(a) The utility function without activation costs; (b) aggregate reserved rate for users of slices; (c) the expected outage of downlinks; and (d) the expected acquired revenue for slices by Algorithm \ref{al:parallel2}.}
	\end{figure*}
	
	\begin{proposition}
		The ADMM approach given in Algorithm \ref{al:admm} converges to the global  solution of \eqref{opt:slice2}.
	\end{proposition}
	
	\begin{proof}
	As the quadratic continuous upper-bound (42) has identical value and gradient to the value and gradient of $\Gamma(\mathbf{r})$ in point $\mathbf{r}=\hat{\mathbf{r}}$, the convergence conditions \cite[Assumption 2]{razaviyayn2013unified} are satisfied, and the successive upper-bound (42) minimizations converge to a stationary solution with respect to $\mathbf{r}$, which is also global due to the strong convexity of (41). Moreover, due to Proposition \ref{prop:tir2}, the global solution is obtained by Algorithm 5 for the subproblem with respect to $\{\mathbf{g},\mathbf{t}\}$. Therefore, the obtained solution by the ADMM approach in Algorithm 6 is global stationary \cite[p. 698]{bertsekas1999nonlinear}.
	\end{proof}

	\section{Simulation Results}
	\label{sec:sim}
	In this section, we evaluate the performance of the proposed approaches. The considered network for evaluations is the one given in \cite{Reyhanian1,reyhanian2021resource}, which includes both the backhaul and radio access components. A data center is connected to routers of the network through three gateway routers, GW $1$, GW $2$, and GW $3$. The network includes $57$ APs and $11$ network routers. APs are distributed on the X-Y plane, and they are connected to each other and to routers via wired links. The backhaul network has $162$ links. Wired link capacities are identical in both directions. 
	Backhaul link capacities are determined as
	\begin{itemize}
		\item Links between the data center and routers: $4$ Gnats/s;
		\item Links between routers: 2 Gnats/s;
		\item Links between routers and  APs: $2$ Gnats/s;
		\item 2-hop to the routers: $400$ Mnats/s;
		\item 3-hop to the routers: $320$ Mnats/s;
		\item 4-hop to the routers: $160$ Mnats/s.
	\end{itemize}
	The considered paths originate from the data center and are extended toward users. User-AP associations are determined by the highest long-term received power. We consider three wireless connections that have the highest received power to serve each user. There are three paths for carrying each user's data from a data center to APs. The distribution of the demand is log-normal:
	\begin{align}
		d_k \sim \frac{1}{d_k \sigma_k \sqrt{2\pi}}\exp(-\frac{(\ln d_k-\eta_k)^2}{2 \sigma_k^2}).\nonumber
	\end{align} 
	In addition, it is assumed that $\eta_k$ is realized randomly from a normal distribution for each user. The power allocations in APs are fixed. The dispensed resource in an AP is bandwidth. The channel between each user and an AP is a Rayleigh fading channel. The CDF of the wireless channel capacity, which is parameterized by the allocated bandwidth $t_k^p$, is given as \cite{4411539}
	\begin{align}
		1-\exp(\frac{1-2^{v_k^p/t_k^p}}{\overline{\text{SNR}_k^p}}),\nonumber
	\end{align}
	where $\overline{\text{SNR}_k^p}$ is the average SNR in one time-scale. The PDF of the wireless channel capacity is
	\begin{align}
		\frac{\ln(2)2^{v_k^p/t_k^p}\exp(\frac{1-2^{v_k^p/t_k^p}}{\overline{\text{SNR}_k^p}})}{\overline{\text{SNR}_k^p} t_k^p}.\nonumber
	\end{align}

	\begin{table*}
		\centering
		\scalebox{0.9}{
			\begin{tabular}{ |p{1.cm}|p{1cm}|p{0.8cm}|p{1cm}|p{0.5cm}|p{0.5cm}|p{1.4cm}|p{0.4cm}|p{1.2cm}|p{0.5cm}|p{0.5cm}|}
				\hline\multicolumn{1}{|c|}{AP Budget} &\multicolumn{5}{|c|}{Scenario $1$} &\multicolumn{5}{|c|}{Scenario $2$}   \\\hline
				\multicolumn{1}{|c|}{}&\multicolumn{2}{|c|}{Optimal solution} &\multicolumn{3}{|c|}{Algorithm \ref{al:parallel2}}&\multicolumn{2}{|c|}{Optimal solution} &\multicolumn{3}{|c|}{Algorithm \ref{al:parallel2}}\\\hline
				\multicolumn{1}{|c|}{}&\multicolumn{1}{|c|}{Activated slices}& \multicolumn{1}{|c|}{Obj. fun.}&\multicolumn{1}{|c|}{Activated slices}&\multicolumn{1}{|c|}{Obj. fun.}&\multicolumn{1}{|c|}{App. ratio}&\multicolumn{1}{|c|}{Activated slices}& \multicolumn{1}{|c|}{Obj. fun.}&\multicolumn{1}{|c|}{Activated slices}&\multicolumn{1}{|c|}{Obj. fun.}&\multicolumn{1}{|c|}{App. ratio}\\\hline
				20 MHz& 1, 4 & 1,446 & 1, 4 & 1,446 & 1 & 1, 4 & 1,655 & 1, 4 & 1,655 & 1\\\hline	
				35 MHz& 1, 4 & 2,055 & 1, 4 & 2,055  &  1& 1, 4 & 2,401 & 1, 4 & 2,401 & 1\\\hline	
				50 MHz& 1, 4, 5 & 2,358 & 1, 4 & 2,331 & 0.988 & 1, 4, 5 & 3,066 & 1, 3, 4 & 3,000 & 0.978\\\hline	
				65 MHz& 1, 4, 5 & 2,625 & 1, 4, 5 & 2,625 &  1 & 1, 2, 4 & 3,436 & 1, 3, 4, 5 & 3,386 & 0.985\\\hline	
				80 MHz& 1, 2, 4, 5 & 2,863 & 1, 4, 5 & 2,813 & 0.982 & 1, 2, 3, 4 & 3,852 &  1, 2, 4, 5 & 3,852 & 0.999\\\hline	
		\end{tabular}}
		\caption{The performance of Algorithm \ref{al:parallel2} against the optimal solution. }\label{tab:cpu2}
	\end{table*}

	Suppose that $200$ users are served by five different slices of a single tenant. We consider $\psi_j=1$ and $\theta_s=3$, and we increase the AP bandwidth budget from $20$ MHz to $80$ MHz with a step size of $15$ MHz. The minimum required reserved rate for the slices is $600$ Mnats/s. In the considered setup, $\phi_{k,L_n}(r_k)=90-\exp(-0.045r_k+4.5)$, where $r_k$ is in Mnats/s. The cost of activating a slice is $1000$ units of utility, and $q=0.1$ and $\epsilon=0.05$. We consider four possible $\mathcal{K}_s$ for each slice. Furthermore, we consider two scenarios. 
	We list the considered scenarios as follows:
	\begin{enumerate}
		\item Scenario 1: $\mathbb{E}[\eta_k]$ for users of slices $1$ and $4$ is $3.5$ Mnats/s, and $\mathbb{E}[\eta_k]$ for users of slices $2$, $3$, and $5$ is $2.5$ Mnats/s;
		\item Scenario 2: $\mathbb{E}[\eta_k]$ for users of slices $1$ and $4$ is $4$ Mnats/s, and $\mathbb{E}[\eta_k]$ for users of slices $2$, $3$, and $5$ is $3$ Mnats/s.
	\end{enumerate}
	In Fig. \ref{fig:act_obj}, we plot the gained utility by each slice given in \eqref{eq:util} without the cost of activation. When the depicted gained utility from a slice is zero in Fig. \ref{fig:act_obj}, that slice is not activated. It is observed that in both scenarios, the number of activated slices increases as the bandwidth budget of APs increases. The reason is that when the available bandwidth in APs increases, the expected outage of downlinks (Fig. \ref{fig:act_outage}) decreases, and therefore, the overall expected gained utilities of slices increase. Furthermore, we observe from Fig. \ref{fig:act_obj} that as users served by slices $1$ and $4$ request greater data rates, these slices, which provide higher revenue, are activated first. 
	
	
	From Fig. \ref{fig:act_rate}, we observe that the aggregate reserved rates for users of each slice increases as the bandwidth budget of APs increases if the number of active slices does not increase. When the number of active slices increases, the amount of allocated resources to users of a slice becomes smaller and reserved rates for users of a slice can decrease.  We observe from Fig. \ref{fig:act_outage} that the expected outage of users decreases as the bandwidth budget of APs increases. This helps the slices to improve their utility functions. The collected revenue of each slice from serving users is depicted in Fig. \ref{fig:act_rev}. We compare the solutions returned by Algorithm \ref{al:parallel2} against the optimal ones, obtained by an exhaustive search, in Table \ref{tab:cpu2}. According to Table \ref{tab:cpu2}, the obtained objective functions (including the activation costs) from Algorithm \ref{al:parallel2} are identical to their optimal values in six out of ten different experiments. The lowest approximation ratio of a solution obtained by Algorithm \ref{al:parallel2} is $0.978$, which clearly demonstrates the efficiency of Algorithm \ref{al:parallel2}. 
		\subsection{Slice Activation}
	\begin{figure*}[t!]
		\centering
		\begin{subfigure}[t]{0.35\textwidth}
			\includegraphics[width=\textwidth]{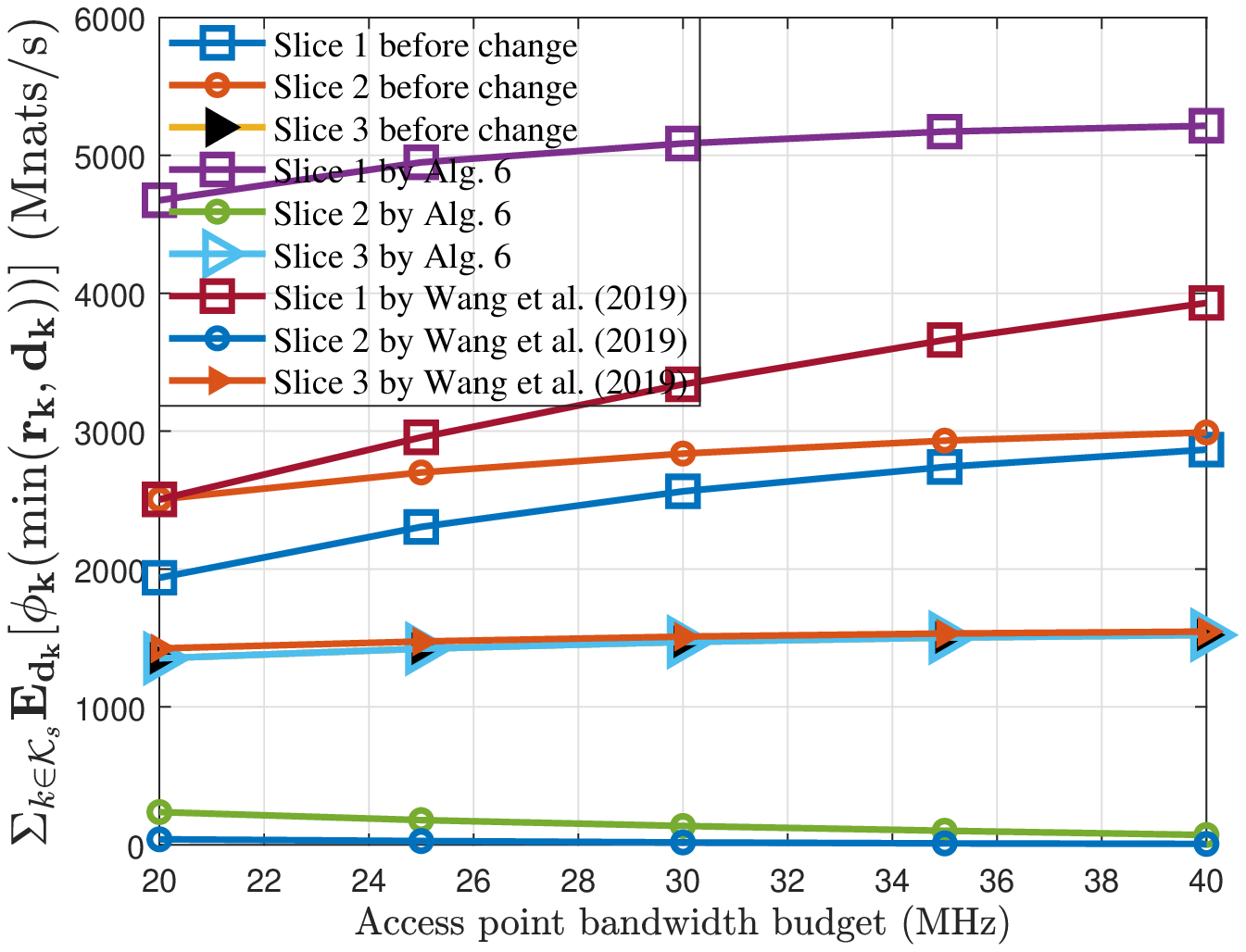}\caption{}\label{fig:chan_rev}
		\end{subfigure}
		\begin{subfigure}[t]{0.35\textwidth}
			\includegraphics[width=\textwidth]{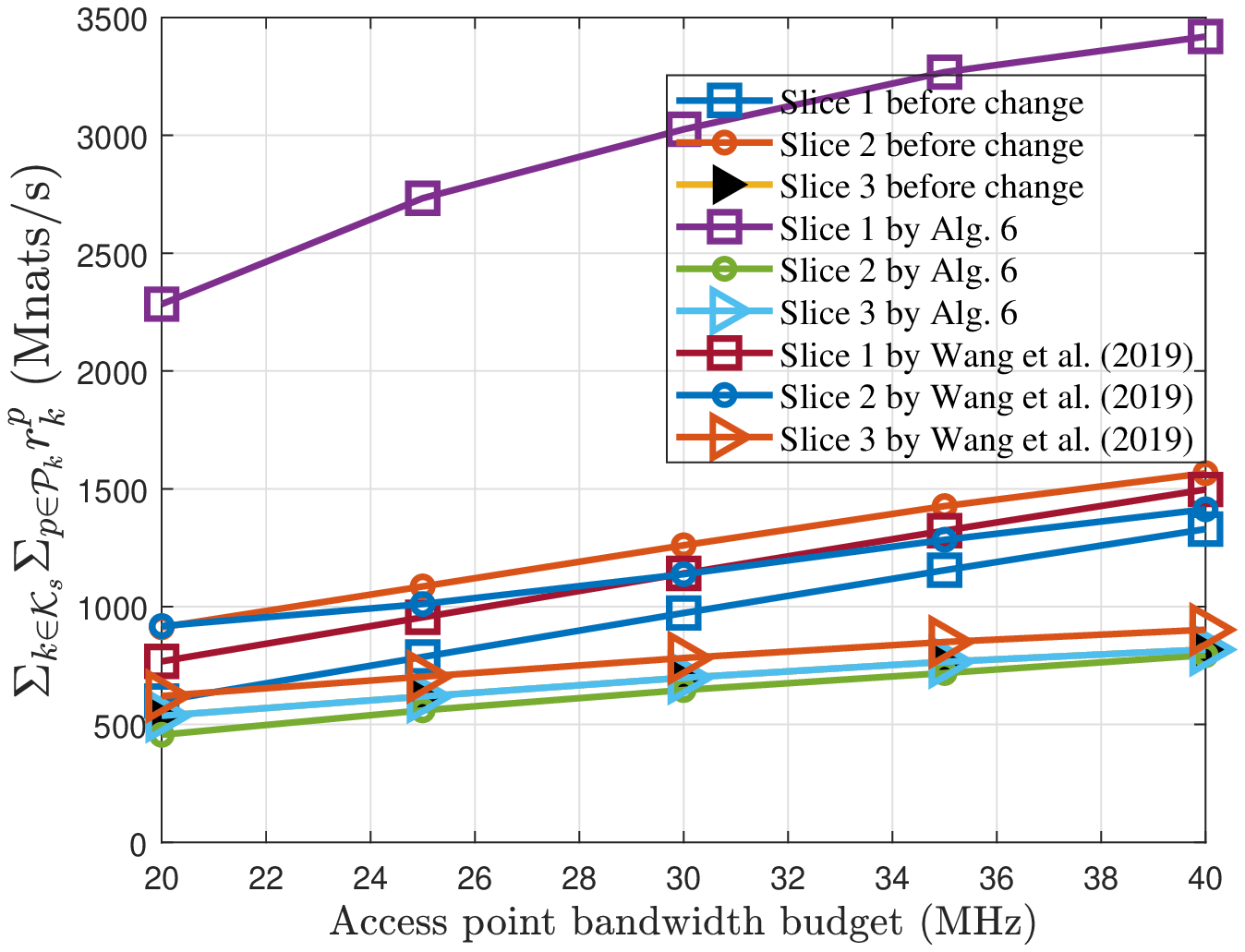}\caption{}\label{fig:chan_rate}
		\end{subfigure}\\
		\begin{subfigure}[t]{0.35\textwidth}
			\includegraphics[width=\textwidth]{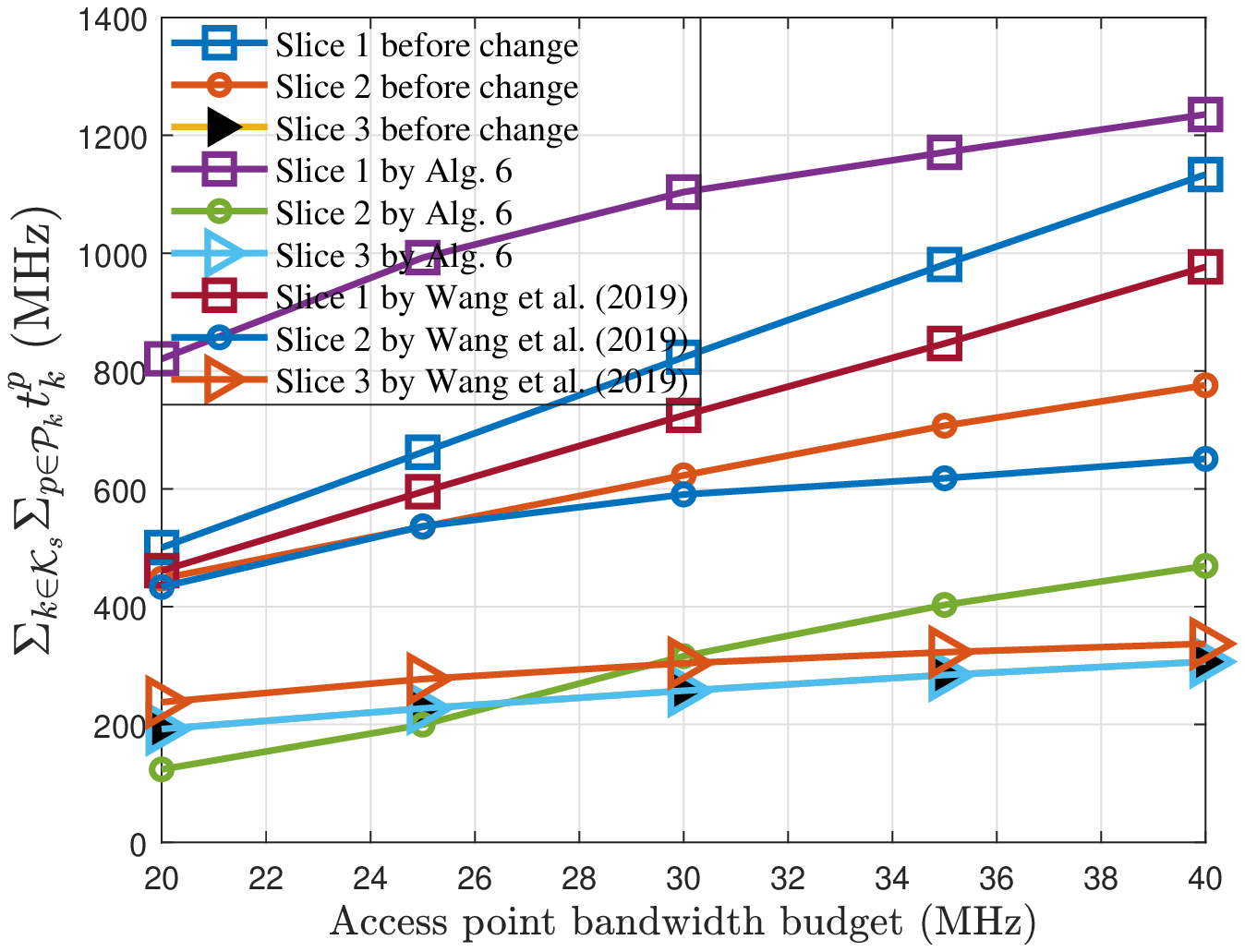}\caption{}\label{fig:chan_ban}
		\end{subfigure}
		\begin{subfigure}[t]{0.35\textwidth}
			\includegraphics[width=\textwidth]{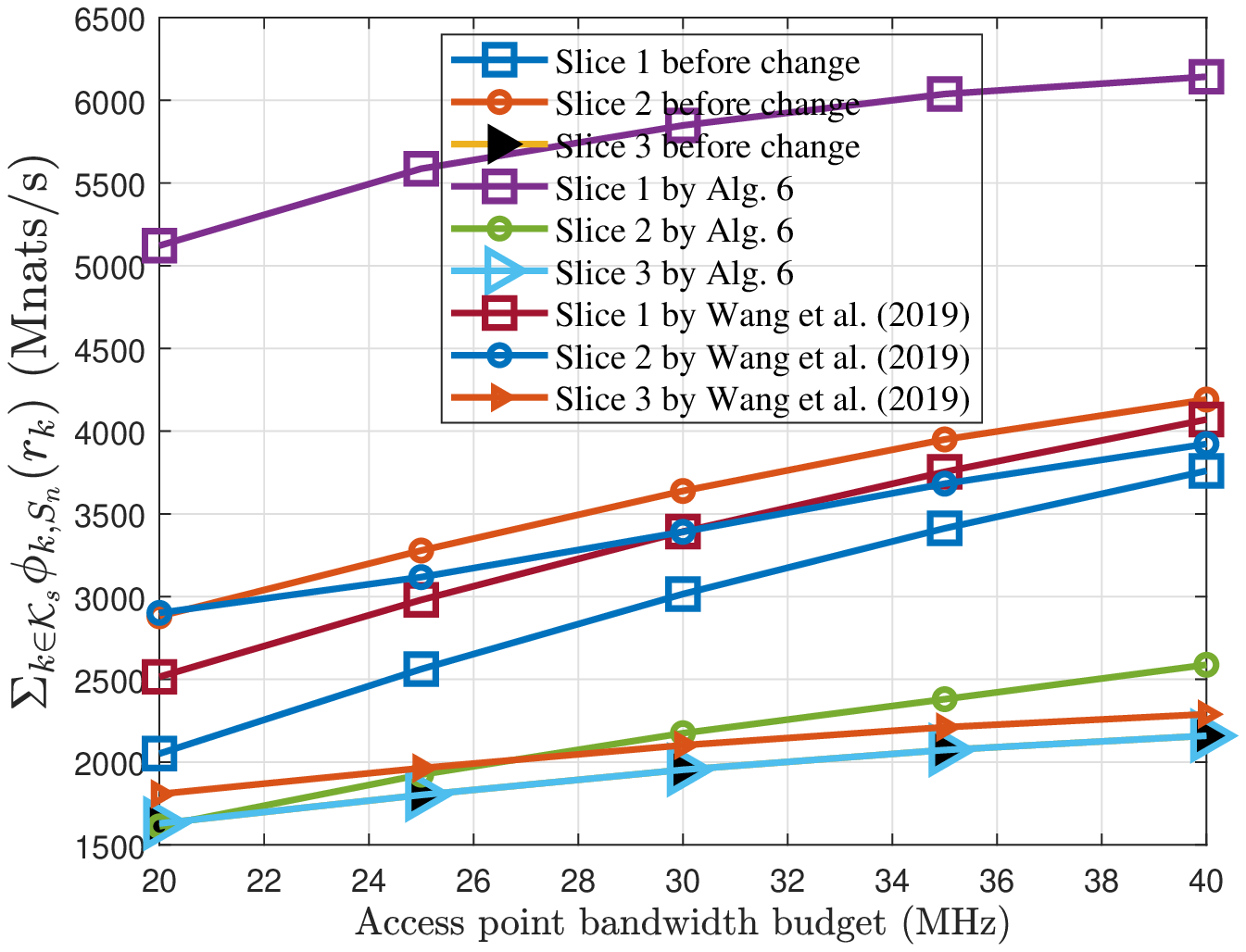}\caption{}\label{fig:max_rev}
		\end{subfigure}
		\caption{(a)  Expected gained revenue by slices; (b) reserved rates for slices; (c) reserved bandwidth for slices; (d) the maximum revenue that can be obtained for slices by Algorithm \ref{al:admm}.}
	\end{figure*}
	
	In the above experiments, the number of required iterations for the convergence of the Frank-Wolfe algorithm is at most $60$. Furthermore, Algorithm \ref{al:parallel2} converges after at most $30$ iterations. The CPU time for each test (which corresponds to one scenario and an AP budget) in our C implementation is less than $14$ seconds.
	\subsection{Slice Reconfiguration}
	We consider that $200$ users are served by three different tenants, where each owns a single slice. The first and second tenants serve $80$ users each. The third tenant serves 40 users. We consider $\psi_j=1$ and increase the AP bandwidth budget from $20$ MHz to $40$ MHz with a step size of $5$ MHz.The minimum required reserved rates for the slices are $500$ Mnats/s, $500$ Mnats/s and $400$ Mnats/s, respectively. In the considered setup, $\boldsymbol{\beta}=[0.1,0.2,0.3]$ and $\phi_{k,S_n}(r_k)=90-\exp(-0.09r_k+4.5)$, where $r_k$ is in Mnats/s. We consider $\mathbb{E}[\eta_k]=2.5$ Mnats/s for all users. Before the arrival of new statistics, we observe from Fig. \ref{fig:chan_rev} that the obtained revenue by slice $2$ is greater than that by slice $1$. The reason is that $\beta_2>\beta_1$. Thus, providing a certain rate to a user in slice $2$ requires less bandwidth. As a result, more rates are reserved for slice $2$, as it is depicted in Fig. \ref{fig:chan_rate}, and less bandwidth reserved for slice $2$, as depicted in Fig. \ref{fig:chan_ban}. We observe that the expected revenue for each slice enhances with the increase of bandwidth budget.  Consider that the traffic statistics of each user served by the first and the second slice change dramatically as follows:
	\begin{itemize}
		\item For the users served by the first slice, the mean of $\eta_k$ becomes $3.5$ Mnats/s. 
		\item The distribution of the demand for each user served by slice $2$ becomes exponential with a mean of $2$ Mnats/s. 
	\end{itemize} 
	The PDF of the demand for each user supported by the third slice remains identical. Furthermore, $\boldsymbol{\beta}$ changes to $\boldsymbol{\beta}=[0.1,0.6,0.3]$.
	As $\eta_k$ increases, the expected aggregate demanded traffic of users served by the first slice increases compared to its value in the first place. Furthermore, the expected aggregate demanded traffic of users served by the second slice diminishes compared to its previous value. Consider that the minimum reserved rate constraints of the first and second slices are relaxed after the change. However, the minimum reserved rate constraint of the third slice remains $400$ Mnats/s.  It is observed from Fig. \ref{fig:chan_rev} that the expected obtained revenue by the third slice does not change due to the enforced sparsity by the $\ell_0-$norm. 
	Numerical simulation confirms that after three updates, coefficients $a_s^{1,i}$ and $a_s^{2,i}$ converge. To solve each group LASSO subproblem, at most $15$ iterations of the outer ADMM in Algorithm \ref{al:admm} with $7$ iterations of the inner ADMM in Algorithm \ref{al:gt} are required. The overall CPU time is less than $27$ seconds in our C implementation.

	We compare Algorithm \ref{al:admm} against the one proposed in \cite{wang2019reconfiguration}. In \cite{wang2019reconfiguration}, $\ell_1$-norm is used to promote sparsity in slice reconfigurations. To tackle $\ell_1$-norm non-differentiability, Wang \textit{et al.} introduced affine constraints that allow limited slice variations instead of keeping $\ell_1$-norm in the objective function. We allow $1$ Mnats/s variation for the reserved rate for each path and $1$ MHz variation for the reserved bandwidth for each path. Under the same bandwidth budgets in APs, we observe from Figs. \ref{fig:chan_rev}, \ref{fig:chan_rate} and \ref{fig:chan_ban} that
	our approach is able to better reconfigure slices to achieve a higher expected revenue. Moreover, it reserves higher rates and more resources for users.  In Fig. \ref{fig:max_rev}, we depict the maximum revenue that can be obtained by slices. This takes place when the demanded rate by each user is at least equal to the reserved rate for that user. We observe that the maximum obtained revenue by slices using the reserved resources in the backhaul and RAN via Algorithm \ref{al:admm} is greater than that by \cite{wang2019reconfiguration}.

	\section{Concluding Remarks}\label{sec:con}
	In this paper, we studied jointly slicing link capacity in the backhaul and transmission resources in RAN for multiple network tenants prior to the observation of user demands. We proposed a novel two time-scale framework for the activation of network slices and also reconfiguring active slices based on the time varying statistics from user demands and channel states.
	We proposed $\ell_q$-norm, $0<q<1$, regularization to promote sparsity in the activation of network slices. Due to the non-convexity of the formulated activation problem, we successively solved a sequence of convex approximations of the problem via a novel Frank-Wolfe algorithm. Furthermore, we formulated the slice reconfiguration problem and since the reconfiguration of network slices can be costly, we used group LASSO regularization to enhance the sparsity of reconfigurations for active slices. An efficient, distributed and parallel algorithm is proposed to solve each group LASSO subproblem.
	Through extensive numerical tests, we verified the efficacy and efficiency of our approaches against the optimal solutions and existing state of the art method.

	\ifCLASSOPTIONcaptionsoff
	\newpage
	\fi

	\bibliographystyle{IEEEbib}
	\bibliography{ref_SDRA}

\end{document}